%% file: main_focs2014_dfa.tex
\documentclass[11pt]{article}

\usepackage{fullpage}

\usepackage{amsfonts}
\usepackage{amssymb}
\usepackage{amstext}
\usepackage{amsmath}
\usepackage{mathtools}

\usepackage{tikz}
\usetikzlibrary{calc}
\usetikzlibrary{shapes,arrows,positioning,shadows,snakes}

\usepackage{footmisc}

\usepackage{amsthm} 

\usepackage{caption}
\usepackage{subcaption}

\usepackage{color}
\usepackage{nameref}
\definecolor{ForestGreen}{rgb}{0.1333,0.5451,0.1333}
\definecolor{DarkRed}{rgb}{0.8,0,0}
\definecolor{Red}{rgb}{1,0,0}

\usepackage{thmtools} 

\usepackage{graphicx}
\usepackage{graphics}
\usepackage{colordvi}
\usepackage{xspace}
\usepackage{algorithm}
\usepackage{algorithmicx}
\usepackage{algpseudocode}
\usepackage{url}
\usepackage{enumitem}

\newcommand{\dfa}{\ensuremath{{\sf dfa}}\xspace} 
 
\newcommand{\cnf}{\ensuremath{{\sf cnf}}\xspace}

\newcommand{\DFA}{\ensuremath{{\sf DFA}}\xspace} 
\newcommand{\NFA}{\ensuremath{{\sf NFA}}\xspace}

        {\hspace*{\fill}$\Box$\par\vspace{4mm}}

\newcommand{\enc}[1]{\langle #1 \rangle}

\makeatletter
\def\thmt@refnamewithcomma #1#2#3,#4,#5\@nil{%
  \@xa\def\csname\thmt@envname #1utorefname\endcsname{#3}%
  \ifcsname #2refname\endcsname
    \csname #2refname\expandafter\endcsname\expandafter{\thmt@envname}{#3}{#4}%
  \fi
}
\makeatother

\declaretheorem[numberwithin=section,refname={Theorem,Theorems},Refname={Theorem,Theorems}]{theorem}
\declaretheorem[numberlike=theorem,refname={Lemma,Lemmas},Refname={Lemma,Lemmas}]{lemma}
\declaretheorem[numberlike=theorem,refname={Corollary,Corollaries},Refname={Corollary,Corollaries}]{corollary}
\declaretheorem[numberlike=theorem,refname={Proposition,Propositions},Refname={Proposition,Propositions}]{proposition}

\declaretheorem[numberlike=theorem,refname={observation,observations},Refname={Observation,Observations}]{observation}

\declaretheorem[numberlike=theorem,refname={Claim, Claims},Refname={Claim, Claims}]{claim}

\declaretheorem[numberlike=theorem]{definition}

\newcommand{\yi}{{\sc Yes-Instance}\xspace}
\renewcommand{\ni}{{\sc No-Instance}\xspace}

\newcommand{\ceil}[1]{\ensuremath{\left\lceil#1\right\rceil}}

\newcommand{\paren}[1]{\left ( #1 \right ) }

\newcommand{\polylog}[1]{\mathrm{polylog}(#1)}

\newcommand{\opt}{\mbox{\sf OPT}}

\newcommand{\set}[1]{\left\{ #1 \right\}}

\newcommand{\pset}{{\mathcal{P}}}

\newcommand{\aset}{{\mathcal{A}}}

\newcommand{\hset}{{\mathcal{H}}}
\newcommand{\nset}{{\mathcal{N}}}
\newcommand{\fset}{{\mathcal{F}}}

\newcommand{\be}{\begin{enumerate}}
\newcommand{\ee}{\end{enumerate}}
\newcommand{\bd}{\begin{description}}
\newcommand{\ed}{\end{description}}
\newcommand{\bi}{\begin{itemize}}
\newcommand{\ei}{\end{itemize}}

\renewcommand{\phi}{\varphi}

\newcommand{\poly}{\operatorname{poly}}

\newcommand{\wall}{\overrightarrow{R}}
\newcommand{\awall}{\widehat{R}}

\newcommand{\edp}{{\sf edp}\xspace}
\newcommand{\tildeEDP}{\widetilde{\sf edp}\xspace}

\newcommand{\VDP}{{\sf VDP}\xspace}
\newcommand{\EDP}{{\sf EDP}\xspace}

\ifdefined\ShowComment

\def\danupon#1{\marginpar{$\leftarrow$\fbox{D}}\footnote{$\Rightarrow$~{\sf #1 --Danupon}}}
\def\danuponb#1{{\bf Danupon:} #1}
\def\parinya#1{\marginpar{$\leftarrow$\fbox{P}}\footnote{$\Rightarrow$~{\sf #1 --Parinya}}}
\def\bundit#1{\marginpar{$\leftarrow$\fbox{B}}\footnote{$\Rightarrow$~{\sf #1 --Bundit}}}

\else

\def\danupon#1{}
\def\danuponb#1{}
\def\parinya#1{}
\def\bundit#1{}

\fi

\title{Pre-Reduction Graph Products: Hardnesses of Properly Learning DFAs and Approximating EDP on DAGs}

\date{}

\author{ 
	      Parinya Chalermsook\thanks{Max-Planck-Institut f\"ur Informatik, Germany. 
          Work partially done while at IDSIA, Switzerland.
          Supported by the Swiss National Science Foundation project 200020\_144491/1}
 \and
        Bundit Laekhanukit\thanks{McGill University, Canada. 
          Supported by the Natural Sciences and
          Engineering Research Council of Canada (NSERC) grant
          no.~28833 and by Dr\&Mrs M.Leong fellowship. Istituto Dalle Molle di Studi sull'Intelligenza Artificiale (IDSIA). Supported by ERC starting grant 279352 (NEWNET).} 
 \and
        Danupon Nanongkai\thanks{University of Vienna, Austria. Work partially done while at Nanyang Technological University (NTU), Singapore, and ICERM, Brown University, USA.}      
}

\begin{document}


\begin{titlepage}
\maketitle
\pagenumbering{roman}
\input{abstract}

\newpage
\setcounter{tocdepth}{2}
\tableofcontents
\end{titlepage}

\newpage
\pagenumbering{arabic}

\input{intro}

\input{overview}

\input{prelim}
\input{meta}

\input{dfa-product} 
\input{dfa-tight} 

\input{pac-learning}

\input{edp-product}

\input{others}

\input{conclusion}

\bibliographystyle{alpha}
\bibliography{reference}

\newpage
\appendix
\section*{Appendix}

\input{bad-examples}

\end{document}

%% file: abstract.tex

\begin{abstract}
The study of graph products is a major research topic  and typically concerns the term $f(G*H)$, e.g., to show that $f(G*H)=f(G)f(H)$. In this paper, we study graph products in a non-standard form $f(R[G*H]$ where $R$ is a ``reduction'', a transformation of any graph into an instance of an intended optimization problem. We resolve some open problems as applications.

The first problem is {\em minimum consistent deterministic finite automaton (DFA)}. We show a tight $n^{1-\epsilon}$-approximation hardness, improving the $n^{1/14-\epsilon}$ hardness of [Pitt and Warmuth, STOC~1989 and JACM 1993], where $n$ is the sample size. (In fact, we also give improved hardnesses for the case of {\em acyclic} DFA and NFA.) Due to Board and Pitt [Theoretical Computer Science 1992], this implies the {\em hardness of properly learning DFAs} assuming $NP\neq RP$ (the weakest possible assumption). This affirmatively answers an open problem raised 25 years ago in the paper of Pitt and Warmuth and the survey of Pitt [All 1989]. Prior to our results, this hardness only follows from the stronger hardness of {\em improperly} learning DFAs, which requires stronger assumptions, i.e., either a cryptographic or an average case complexity assumption [Kearns and Valiant STOC~1989 and J.~ACM 1994; Daniely~et~al. STOC 2014]. 
The second problem is {\em edge-disjoint paths} (EDP) on {\em directed acyclic graphs} (DAGs). This problem admits an $O(\sqrt{n})$-approximation algorithm [Chekuri, Khanna, and Shepherd, Theory of Computing 2006] and a matching $\Omega(\sqrt{n})$ integrality gap, but so far only an $n^{1/26-\epsilon}$ hardness factor is known  [Chuzhoy~et~al., STOC~2007].  ($n$ denotes the number of vertices.) Our techniques give a tight $n^{1/2-\epsilon}$ hardness for EDP on DAGs, thus resolving its approximability status. 

As by-products of our techniques: (i) We give a tight hardness of packing vertex-disjoint $k$-cycles for large $k$, complimenting [Guruswami and Lee, ECCC~2014] and matching [Krivelevich~et~al., SODA~2005 and ACM~Transactions~on~Algorithms~2007]. (ii) We give an alternative (and perhaps simpler) proof for the hardness of properly learning DNF, CNF and intersection of halfspaces [Alekhnovich~et~al., FOCS 2004 and J.~Comput.Syst.~Sci.~2008]. Our new concept reduces the task of proving hardnesses to merely analyzing graph product inequalities, which are often as simple as textbook exercises. This concept was inspired by, and can be viewed as a generalization of, the {\em graph product subadditivity} technique we previously introduced in SODA~2013. This more general concept might be useful in proving other hardness results as well.
\end{abstract}

%% file: intro.tex
\section{Introduction}\label{sec:intro}


\subsection{The Concept of Pre-Reduction Graph Product}

\paragraph{Background: Graph Product and Hardness of Approximation.}
Graph product is a fundamental tool with rich applications in both graph theory and theoretical computer science. 
It is, roughly speaking, a way to combine two graphs, say $G$ and $H$, into a new graph denoted by $G*H$. 
For example, the following {\em lexicographic product}, denoted by $G\cdot H$, will be particularly useful in this paper. 
\begin{align}
\mbox{(Lexicographic Product)} && V(G\cdot H) & = V(G)\times V(H)  =\{(u, v) : u\in V(G) \mbox{ \underline{and} } v\in V(H)\}.\nonumber\\ 
&&E(G\cdot H) &= \{(u, a)(v, b) : \mbox{$uv \in E(G)$ \underline{or} {\bf (}$u=v$ \underline{and} $ab\in E(H)${\bf )}}\}.\label{eq:lex product definition}
\end{align}

%
%
A common study of graph product aims at understanding how $f(G*H)$ behaves for some function $f$ on graphs denoting a graph property. 
For example, if we let $\alpha(G)$ be the {\em independence number} of $G$ (i.e., the cardinality of the maximum independent set),
then $\alpha(G\cdot H) = \alpha(G)\alpha(H).$ 

Graph products have been extremely useful in {\em boosting} the hardness of approximation. One textbook example is proving the hardness of $n^{\epsilon}$ for approximating the maximum independent set problem (i.e., approximating $\alpha(G)$ of an input graph $G$): Berman and Schnitger~\cite{BermanS92} showed that we can reduce from Max 2SAT to get a constant approximation hardness $c>1$ for the maximum independent set problem, and then use a graph product to boost the resulting hardness to $n^{\epsilon}$ for some (small) constant $\epsilon$.
To illustrate how graph products amplify hardness, suppose we have a $(1.001)$-gap reduction $R[I]$ that transforms an instance $I$ of SAT into a graph $G$. 
Since $\alpha(\cdot)$ is multiplicative, if we take a product $R[I]^k$ for any integer $k$, the hardness gap immediately becomes $(1.0001)^k = 2^{\Omega(k)}$. Choosing $k$ to be large enough gives $2^{\log^{1-\epsilon} n}$ hardness.
Therefore, once we can rule out the PTAS, graph products can be used to boost the hardness to almost polynomial. 
This idea is also used in many other problems, e.g., in proving the hardness of the longest path problem~\cite{KargerMR97}.  

\paragraph{Our Concept: Pre-Reduction Graph Product.} This paper studies a reversed way to apply graph products: instead of the commonly used form of $(R[I])^k = (R[I]*R[I]*\ldots )$  to boost the hardness of approximation, we will use $R[I^k]=R[I*I*\ldots]$; here, $I$ is a graph which is an instance of a hard graph problem such as maximum independent set or minimum coloring.  We refer to this approach as {\em pre-reduction graph product} to contrast the previous approach in which graph product is performed {\em after} a reduction (which will be referred to as {\em post-reduction} graph product). 
The main conceptual contribution of this paper is the demonstration to the power and versatility of this approach in proving approximation hardnesses.  We show our results in Section~\ref{sec:results} and will come back to explain this concept in more detail in Section~\ref{sec:overview}. 

%
We note one conceptual difference here between the previous post-reduction and our pre-reduction approaches: While the previous approach starts from a reduction $R$ that already gives some hardness result, our approach usually starts from a reduction that does {\em not immediately} provide any hardness result; in other words, such reduction alone cannot be used to even prove NP-hardness. (See Section~\ref{sec:overview} for an illustration.)
Moreover, in contrast to the previous use of $(R[I])^k$ which requires $R[I]$ to be a graph, our approach allows us to prove hardnesses of problems whose input instances are not graphs.
Also note that our approach gives rise to a study of graph products in a new form: in contrast to the usual study of $f(G*H)$, our hardness results crucially rely on understanding the behavior of $f(R[G*H])$ for some function $f$, reduction $R$, and graph product $*$ (which happens to always be the lexicographic product in this paper).
Another feature of this approach is that it usually leads to simple proofs that do not require heavy machineries (such as the PCP-based construction) -- some of our hardness proofs are arguably simplifications of the previous ones; in fact, most of our hardness results follow from the meta-theorem (see Section~\ref{sec:meta}) which shows that a bounds of $f(R[G*H])$ in a certain form will immediately lead to hardness results. We list some bounds of $f(R[G*H])$ in Theorem~\ref{thm:intro bounds of product}.

\subsection{Problems and Our Results}\label{sec:results}

\begin{table}
\centering
\begin{tabular}{c|c|c|c}
\hline
Problems & Upper Bounds &  Prev. Hardness & New Hardness\\
\hline
{\sc MinCon($DFA$, $DFA$)} & $O(n)$ &  $n^{1/14-\epsilon}$ \cite{PittW93} & $n^{1-\epsilon}$\\
\EDP on DAGs & $\tilde O(n^{1/2})$  \cite{ChekuriKS06} &  $n^{1/26-\epsilon}$ \cite{ChuzhoyGKT07} & $n^{1/2-\epsilon}$\\
$k$-cycle packing & $O(\min(k, n^{1/2}))$ & $\Omega(k)$ \cite{GuruswamiL14} & $O(\min(k, n^{1/2-\epsilon}))$\\
\hline
MinCon($CNF$, $CNF$), & & \\
MinCon($DNF$, $DNF$), & $O(n)$ & $n^{1-\epsilon}$ & $n^{1-\epsilon}$ \\
MinCon(Halfspace,Halfspace) & & & {\small (Alternative proof)} \\
\hline
\end{tabular}
\caption{Summary of our hardness results.}\label{table:results}
\end{table}

\subsubsection{Minimum Consistent DFA and Proper PAC-Learning DFAs} In the {\em minimum consistent deterministic finite automaton} (DFA) problem, denoted by {\sc MinCon($DFA$, $DFA$)}, we are given two sets $\pset$ and $\nset$ of {\em positive} and {\em negative sample} strings in $\set{0,1}^*$. We let the {\em sample size}, denoted by $n$, be the total number of bits in all sample strings. Our goal is to construct a DFA $M$ (see Section~\ref{sec:prelim} for a definition) of {\em minimum size} that is {\em consistent} with all strings in $\pset\cup\nset$. That is, $M$ accepts all positive strings $x\in \pset$ and rejects all negative strings $y\in \nset$. 

This problem can be easily approximated within $O(n)$. 
Due to its connections to PAC-learning automata and grammars (e.g.\ \cite{DelaHiguery10-book,Pitt89survey}), the problem has received a lot of attention from the late 70s to the early 90s. The NP-hardness of this problem was proved by Gold \cite{Gold78} and Angluin \cite{Angluin78}. Li and Vazirani \cite{LiV88} later provided the first hardness of approximation result of  $(9/8-\epsilon)$. 
This was greatly improved to $n^{1/14-\epsilon}$ by Pitt and Warmuth \cite{PittW93}.
%
%
Our first result is a tight $n^{1-\epsilon}$ hardness for this problem, improving \cite{PittW93}.
In fact, our hardness result holds even when we allow an algorithm to compare its result to the optimal {\em acyclic} DFA (ADFA), which is larger than the optimal DFA. This problem is called {\sc MinCon($ADFA$, $DFA$)}; see Section~\ref{sec:prelim} for detailed definitions.

\begin{theorem}\label{thm:hardness-mincon-dfa}
Given a pair of positive and negative samples $(\pset,\nset)$ of size $n$ 
where each sample has length $O(\log{n})$, 
for any constant $\epsilon>0$, 
it is NP-hard to distinguish between the following two cases of MinCon(ADFA,DFA):
\begin{itemize}
\item {\sc Yes-Instance:} There is an ADFA of size $n^{\epsilon}$ consistent with $(\pset,\nset)$.
\item {\sc No-Instance:} Any DFA that is consistent with $(\pset,\nset)$ has size at least $n^{1-\epsilon}$. 
\end{itemize}
In particular, it is NP-hard to approximate the minimum consistent DFA problem to within a factor of $n^{1-\epsilon}$.
\end{theorem}


The main motivation of this problem is its connection to the notion of {\em properly PAC-learning} DFAs. It is one of the most basic problems in the area of proper PAC-learning~\cite{DelaHiguery10-book,Pitt89survey,PittW93}. 
Roughly speaking, the problem is to learn an unknown DFA $M$ from given random samples, where a learner is asked to output (based on such random samples) a DFA $M'$ that closely approximates $M$ (see, e.g., \cite{Feldman08} for details). 
The main question is whether DFA is properly PAC-learnable. 

This question was the main motivation behind \cite{PittW93}; however, the $n^{1/14-\epsilon}$ hardness in \cite{PittW93} was not strong enough to prove this. 
Kearns and Valiant \cite{KearnsV94} showed that a proper PAC-learning of DFAs is not possible if we assume a cryptographic assumption stronger than $P\neq NP$. 
In fact, their result implies that even {\em improperly} PAC-learning DFAs (i.e., the output does not have to be a DFA) is impossible. 
Very recently, Daniely~et~al.~\cite{DanielyLS14} obtained a similar result by assuming a (fairly strong) average-case complexity assumption generalizing Feige's assumption~\cite{Feige02}. 
%


The question whether the cryptographic assumption could be replaced by the $RP\neq NP$  assumption (which would be the weakest assumption possible) was asked 25 years ago in \cite{Pitt89survey,PittW93}. In particular, the following is the first open problem in \cite{Pitt89survey}: {\em (i) Can it be shown that DFAs are not properly PAC-learnable based only on the assumption that $RP \neq NP$? (ii) Stronger still, can the improper learnability result of \cite{KearnsV94} be strengthened by replacing the cryptographic assumptions with only the assumption that $RP \neq NP$?}



Applebaum, Barak and Xiao \cite{ApplebaumBX08} showed that proving lower bounds for improper learning using many standard ways of reductions from NP-hard problems 
%
%
will not work unless the polynomial hierarchy collapses, suggesting that an answer to the second question is likely to be negative. For the first question, some hardnesses of proper PAC-learning assuming $RP\neq NP$ were already known at the time (e.g. \cite{PittV88}) and there are many more recent results (see, e.g., \cite{Feldman08} and references therein). 
Despite this, the basic problem of learning DFAs (originally asked in the above question) has remained open. 
Theorem~\ref{thm:hardness-mincon-dfa} together with a result of Board and Pitt~\cite{BoardP92} immediately resolve this problem.


\begin{corollary}\label{cor:pac-learning-dfa}
Unless $\mathrm{NP}=\mathrm{RP}$, the class of DFAs is not properly PAC-learnable.
\end{corollary}

We also note an amusing connection between this type of result and Chomsky's ``Poverty of the Stimulus Argument'', as noted by Aaronson \cite{Aaronson_lecture08}: ``Let's say I give you a list of $n$-bit strings, and I tell you that there's some nondeterministic finite automaton $M$, with much fewer than $n$ states, such that each string was produced by following a path in $M$. Given that information, can you reconstruct $M$ (probably and approximately)? It's been proven that if you can, then you can also break RSA!''
Our Corollary~\ref{cor:pac-learning-dfa} implies that for the case of deterministic finite automaton, being able to reconstruct $M$ will imply not only that one can break RSA but also solve, for instance, traveling salesman problem (TSP) probabilistically.

\subsubsection{Edge-Disjoint Paths on DAGs}\label{sec:intro edp}

In the edge-disjoint paths problem (\EDP) problem, we are given a graph $G = (V,E)$ (which could be directed or undirected) and $k$ source-sink pairs $s_1t_1,s_2t_2, \ldots,s_kt_k$ (a pair can occur multiple times). The objective is to connect as many pairs as possible via edge-disjoint paths. 
%
%
%
Throughout, we let $n$ and $m$ be the number of vertices and edges in $G$, respectively. 
Approximating \EDP has been extensively studied.
It is one of the major challenges in the field of approximation algorithms. 
The problem has received significant attention from many groups of researchers, attacking the problem from many angles and considering a few variants and special cases (see, e.g.,~\cite{RobertsonS95b,Chuzhoy12a,ChuzhoyL12,ChekuriKS09,ChekuriKS05,Kleinberg05,KleinbergT98,KawarabayashiK10} and references therein).  
%

In directed graphs, \EDP can be approximated within a factor of $O(\min(m^{1/2}, n^{2/3}))$ \cite{Kleinberg96,ChekuriK07,VaradarajanV04}. The $O(m^{1/2})$ factor is {\em tight} on sparse graphs since directed \EDP is NP-hard to approximate within a factor of $n^{1/2-\epsilon}$, for any $\epsilon>0$ \cite{GuruswamiKRSY03}. 
In contrast to the directed case, undirected \EDP is much less understood: The approximation factor for this case is $O(n^{1/2})$ \cite{ChekuriKS06} with a matching integrality gap of $\Omega(n^{1/2})$ for its natural LP relaxation, suggesting an $n^{1/2-\epsilon}$ hardness. 
Despite these facts, we only know a $\log^{1/2-\epsilon} {n}$ hardness of approximation assuming $\mathrm{NP} \not \subseteq  \mathrm{ZPTIME}(n^{\polylog n})$.
Even in special cases such as planar graphs (or, even simpler, brick-wall graphs, a very structured subclass of planar graphs), it is still open whether undirected EDP admits an $o(n^{1/2})$ approximation algorithm.  
This obscure state of the art made undirected \EDP one of the  most important, intriguing open problems in graph routing. (Table~\ref{table:edp} summarizes the current status of \EDP.)



One problem that may help in understanding undirected \EDP is perhaps \EDP on {\em directed acyclic graphs} (DAGs). 
This case is interesting because (i) its complexity seems to lie somewhere between the directed and undirected cases, (ii) it shares some similar statuses and structures with undirected \EDP, and (iii) it has close connections to directed cycle packing~\cite{KrivelevichNSYY07} (i.e. hard instances for EDP on DAGs are used as a gadget in constructing the hard instance for directed cycle packing). 
%
%
In particular, on the upper bound side, the technique in \cite{ChekuriKS06} gives an $O(n^{1/2}\poly\log n)$ upper bound not only to undirected \EDP but also to \EDP on DAGs. 
%
%
Moreover, the integrality gap of $\Omega(n^{1/2})$ applies to both cases, suggesting a hardness of $n^{1/2-\epsilon}$ for them. 
However, previous hardness techniques for the case of general directed graphs \cite{GuruswamiKRSY03} completely fail to give a lower bound on both DAGs and undirected graphs\footnote{The result in \cite{GuruswamiKRSY03} crucially relies on the fact that \EDP with $2$ terminal pairs is hard on directed graphs. This is not true if the graph is a DAG or undirected.}.
%
%
On the other hand, subsequent techniques that were invented in \cite{AndrewsZ06,AndrewsCGKTZ10} to deal with undirected \EDP can be strengthened to prove the currently best hardness for DAGs~\cite{ChuzhoyGKT07}\footnote{Their result is in fact proved in a more general setting of \EDP with congestion $c$ for any $c\geq 1$}, which is $n^{1/26-\epsilon}$.  
%
%
%
These results suggest that the complexity of DAGs lies between undirected and directed graphs.  
In this paper, we show that our techniques give a hardness of $n^{1/2-\epsilon}$ for this case, thus completely settling its approximability status. 
Our result is formally stated in the following theorem.\parinya{Changed this paragraph to further avoid aggressiveness. Made points about cycle packing instead.} 

\begin{table}
\centering
\begin{tabular}{c|c|c|c}
\hline
Cases & Upper Bounds & Integrality Gap & Prev. Hardness\\ 
\hline
Undirected & $O(n^{1/2})$ \cite{ChekuriKS06} & $\Omega(n^{1/2})$ & $\log^{1/2-\epsilon} {n}$ \cite{AndrewsCGKTZ10} \\
DAGs & $\tilde O(n^{1/2})$  \cite{ChekuriKS06} & $\Omega(n^{1/2})$ & $n^{1/26-\epsilon}$ \cite{ChuzhoyGKT07}\\
Directed & $O(\min(m^{1/2}, n^{2/3}))$ \cite{Kleinberg96,ChekuriK07,VaradarajanV04} & $\Omega(n^{1/2})$ & $n^{1/2-\epsilon}$~\cite{GuruswamiKRSY03} \\
\hline
\end{tabular}
\caption{The current status of \EDP.}\label{table:edp}
\end{table}

\begin{theorem}\label{thm:hardness-EDP}
Given an instance of \EDP on DAGs, consisting of a graph $G=(V,E)$ on $n$ vertices and a source-sink pairs $(s_1,t_1),\ldots,(s_k,t_k)$,
for any $\epsilon>0$, 
it is NP-hard to distinguish between the following two cases:
\begin{itemize}
\item {\sc Yes-Instance:} There is a collection of edge disjoint paths in $G$ that connects $1/n^{\epsilon}$ fraction of the source-sink pairs. 
\item {\sc No-Instance:} Any collection of edge disjoint paths in $G$ connects at most $1/n^{1/2-\epsilon}$ fraction of the source-sink pairs.
\end{itemize}
In particular, it is NP-hard to approximate \EDP on DAGs to within a factor of $n^{1/2-\epsilon}$.
\end{theorem}




\subsubsection{Other Results}\label{sec:other results}

\paragraph{Minimum Consistent NFA.}
Our techniques also allow us to prove a hardness result for the {\em minimum consistent NFA} problem as stated formally in the following theorem.

\begin{theorem}\label{thm:hardness-mincon-nfa}
Given a pair of positive and negative samples $(\pset, \nset)$ of size $n$ 
where each sample has length $O(\log{n})$, 
for any constant $\epsilon>0$, 
it is NP-hard to distinguish between the following two cases of MinCon(ADFA,NFA):
\begin{itemize}
\item {\sc Yes-Instance:} There is an ADFA of size $n^{\epsilon}$ consistent with $(\pset, \nset)$.
\item {\sc No-Instance:} Any NFA that is consistent with $(\pset, \nset)$ has size at least $n^{1/2-\epsilon}$. 
\end{itemize}
In particular, it is NP-hard to approximate the minimum consistent NFA problem to within a factor of $n^{1/2-\epsilon}$.
\end{theorem}

This improves upon the $n^{1/14 - \epsilon}$ hardness of Pitt and Warmuth~\cite{PittW93}. 
We note that this hardness result is not strong enough to imply a PAC-learning lower bound for NFAs. Such hardness was already known based on some cryptographic or average-case complexity assumptions \cite{KearnsV94,DanielyLS14}. We think it is an interesting open problem to remove these assumptions as we did for the case of learning DFAs.


\paragraph{$k$-Cycle Packing.} 
Our reduction for \EDP can be slightly modified to obtain hardness results for $k$-Cycle Packing, when $k$ is large. 
In the {\em $k$-cycle packing} problem, given an input graph $G$, one wants to pack as many disjoint cycles as possible into the graph while we are only interested in cycles of length at most $k$.  An $O(\min(k, n^{1/2}))$-approximation algorithm for this problem can be easily obtained by modifying the algorithm of Krivelevich~et~al.~\cite{KrivelevichNSYY07}). Very recently, Guruswami and Lee \cite{GuruswamiL14} obtained a hardness of $\Omega(k)$, assuming the Unique Game Conjecture, when $k$ is a constant. This matches the upper bound of Krivelevich~et~al. for small $k$. In this paper, we compliment the result of Guruswami and Lee by showing a hardness of $n^{1/2-\epsilon}$ for some $k\geq n^{1/2}$, matching the upper bound of Krivelevich~et~al. for the case of large $k$.  

\begin{theorem} 
Given a directed graph $G$, for any $\epsilon >0$ and some $k\geq |V(G)|^{1/2}$, it is NP-hard to distinguish between the following cases:
\begin{itemize} 
\item There are at least $|V(G)|^{1/2-\epsilon}$ disjoint cycles of length $k$ in $G$.  

\item There are at most $|V(G)|^{\epsilon}$ disjoint cycles of length at most $2k-1$ in $G$. 
\end{itemize}
In particular, for some $k\geq n^{1/2}$, the $k$-cycle packing problem on $n$-vertex graphs is hard to approximate to within a factor of $n^{1/2-\epsilon}$.  
\end{theorem}

\paragraph{Alternative Hardness Proof for Minimum Consistent CNF, DNF, and Intersections of Halfspaces.}
Our techniques for proving the DFA hardness result can be used to give an alternative proof for the hardness of the minimum consistent DNF, CNF, and intersections of thresholded halfspaces problems.  
In the minimum consistent CNF problem, we are given a collection of samples of size $n$, and our goal is to output a small CNF formula that is consistent with all such samples. Alekhnovich~et~al.~\cite{AlekhnovichBFKP08} previously showed tight hardnesses for these problems, which imply that the classes of CNFs, DNFs, and the intersections of halfspaces are not properly PAC-learnable. Our techniques give an alternative proof (which might be simpler) for these results. 
More specifically, we give an alternative proof for the following theorem and corollary (stated in terms of CNF, but the same holds for DNF and intersection of halfspaces\footnote{It is noted in~\cite{AlekhnovichBFKP08} that one only needs to prove the hardness of CNF, since this problem is a special case of the intersection of thresholded halfspaces problem, and the proof for DNF would work similarly.}). 

\begin{theorem}
Let $\epsilon >0$ be any constant.
Given a pair of positive an negative samples $(\pset, \nset)$ of size $n$ where each sample has length at most $n^{\epsilon}$, it is NP-hard to distinguish between the following two cases: 
\begin{itemize} 
\item \yi: There is a CNF formula of size $n^{\epsilon}$ consistent with $(\pset, \nset)$.  
\item \ni: Any CNF consistent with $(\pset, \nset)$ must have size at least $n^{1-\epsilon}$. 
\end{itemize} 
In particular, it is NP-hard to approximate the minimum consistent CNF problem to within a factor of $n^{1-\epsilon}$.  
\end{theorem}

\begin{corollary} 
Unless $NP=RP$, the class of CNF is not properly PAC-learnable. 
\end{corollary}

%
%
%


%% file: overview.tex
\section{Overview}\label{sec:overview}

%


\subsection{Example of Reduction $R$: Vertex-Disjoint Paths}

To illustrate the pre-reduction graph product concept, consider the {\em vertex-disjoint path} (\VDP) problem. 
The objective of \VDP is the same as that of \EDP except that we want paths to be vertex-disjoint instead of edge-disjoint. The approximability statuses of \EDP and \VDP on DAGs and undirected graphs are the same, and we choose to present \VDP due to its simpler gadget construction. 
Our hardness of \VDP can be easily turned into a hardness of \EDP.


Our goal is to show that this problem has an approximation hardness of $n^{1/2-\epsilon}$, where $n$ is the number of vertices. 
We will use the following reduction\footnote{We thank Julia Chuzhoy who suggested this reduction to us (private communication).} $R$ which transforms a graph $G$ (supposedly an input instance of the maximum independent set problem) into an instance $R[G]$ of the vertex-disjoint paths problem with $\Theta(|V(G)|^2)$ vertices. 
We start with an instance $R[G]$ as in Figure~\ref{fig:intro one} where there are $k$ source-sink pairs (Figure~\ref{fig:intro one} shows an example where $k=6$) and edges are oriented from left to right and from top to bottom. Let us name vertices in $G$ by $1$, $2$, $\ldots$, $k$. 
For any pair of vertices $i$ and $j$, where $i<j$, such that edge $ij$ does {\em not} present in $G$, we remove a vertex $v_{ij}$ from $R[G]$, as shown in Figure~\ref{fig:intro two} (this means that two edges that point to $v_{ij}$ will continue on their directions without intersecting each other). See Section~\ref{sec:edp} for the full description of $R$ in the context of \EDP.  

To see an intuition of this reduction, define a {\em canonical path} be a path that starts at some source $s_i$, goes all the way right, and then goes all the way down to $t_i$ (e.g., a thick (green) path in Figure~\ref{fig:intro two}). 
It can be easily seen that any set of vertex-disjoint paths in $R[G]$ that consists only of canonical paths can be converted to a solution for the maximum independent set problem. 
Conversely any independent set $S$ in $G$ can be converted to a set of $|S|$ vertex-disjoint paths. 
For example, canonical paths between the pairs $(s_1, t_1)$ and $(s_2, t_2)$ in $R[G]$ in Figure~\ref{fig:intro two} can be converted to an independent set $\{1, 2\}$ in $G$ and vice versa.  In other words, if we can {\em force} the \VDP solution to consist only of canonical paths, then we can potentially use the $|V|^{1-\epsilon}$ hardness of maximum independent set to prove a tight $|V|^{1-\epsilon}=|V(R[G])|^{1/2-\epsilon}$ hardness of \VDP. 
This intuition, however, cannot be easily turned into a hardness result since the \VDP solution can use non-canonical paths, and it is possible that $\VDP(R[G])$ is much larger than $\alpha(G)$; see Section~\ref{sec:bad example edp} for an example where $\alpha(G)=O(1)$ and $\VDP(R[G])=\Omega(|V(G)|)$.
Thus, the reduction $R$ by itself cannot be used even to prove that \VDP is NP-hard!

\begin{figure} 
\begin{center} 
\begin{subfigure}[b]{0.4\textwidth}
\centering
               \fbox{\includegraphics[page=1, scale=0.4, clip=true, trim=  11cm 7cm 4.5cm 2cm] {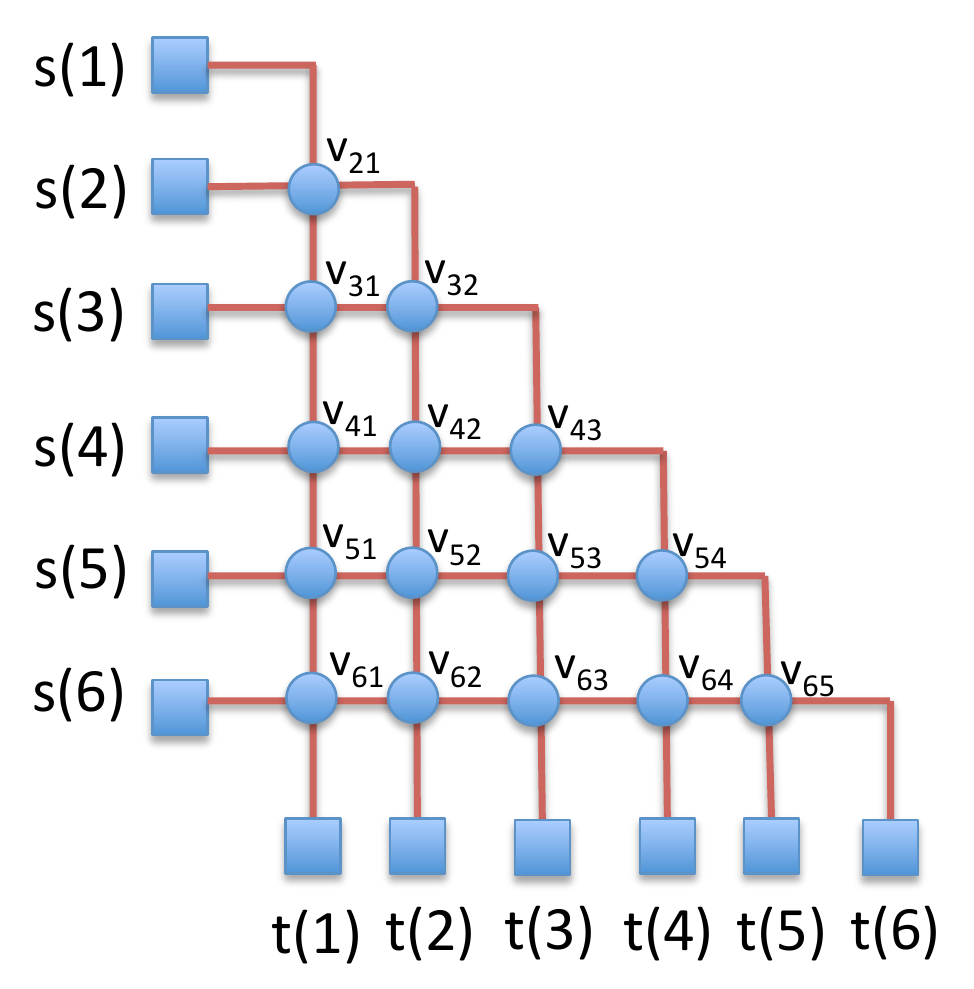}}
               \caption{}
                \label{fig:intro one}
\end{subfigure}%
%
%
\begin{subfigure}[b]{0.6\textwidth}
\centering
				\fbox{\includegraphics[page=3, scale=0.4, clip=true, trim=  1cm 7cm 4.5cm 2cm] {intgap.pdf}}
				\caption{}
                \label{fig:intro two}
\end{subfigure}%
\caption{The reduction $R$ for the vertex-disjoint paths problem. The thick (green) path in Figure~\ref{fig:intro two} shows an example of a canonical path.}
\label{fig:intro}
\end{center} 
\end{figure}

\subsection{The Use of Pre-Reduction Products}
The above situation is very common in attempts to prove hardnesses for various problems. A usual way to obtain hardness results is to modify $R$ into some reduction $R'$. This modification, however, often blows up the size of the reduction, thus affecting its tightness.  
%
For example, \VDP and \EDP on DAGs are only known to be $n^{1/26-\epsilon}$-hard, as opposed to being potentially $n^{1/2-\epsilon}$-hard, as suggested by the integrality gap.
Moreover, the reduction $R'$ is usually much more complicated than $R$. 
%
%
In this paper, we show that for many problems the above difficulties can be avoided by simply picking an appropriate graph product $*$ and understanding the structure of $R[G*G*\ldots]$. To this end, it is sometimes easier to study $f(R[G*H])$ for any graphs $G$ and $H$, although we eventually need only the case where $G=H$. This gives rise to the study of the behavior of $f(R[G*H])$ which is a non-standard form of graph product in comparison with the standard study of $f(G*H)$.
%
%
In fact, most results in this paper follow merely from bounding $f(R[G*H])$ in the form 
\begin{align}
g(G*H)\leq f(R[G*H]) &\leq g(G)f(H)+\poly(|V(G)|), \label{eq:intro gen form}
\end{align}
where $g$ is an objective function of a problem whose hardness is already known (in this paper, $g$ is either maximum independent set or minimum coloring),
and $f$ is an objective function of a problem that we intend to prove hardness. 
Our bounds for functions $f$ corresponding to problems that we want to solve, e.g. the minimum consistent DFA (function $\dfa$) and maximum edge-disjoint paths (function $\edp$)\parinya{We seem to be switching back and forth between VDP and EDP.}, are listed in the theorem below. (Recall that $G\cdot H$ denotes the lexicographic product as defined in Eq.~\ref{eq:lex product definition}.)
\begin{theorem}[Bounds of graph products; informal]\label{thm:intro bounds of product} There is a reduction $R_1$ (respectively $R_2$) that transforms a graph $G$ into an instance of the minimum consistent DFA problem of size $\tilde \Theta(|V(G)|^2)$ (respectively the maximum edge-disjoint paths problem of size $\Theta(|V(G)|^2)$) such that, for any graphs $G$ and $H$,
\begin{align}
\chi(G\cdot H)\leq \dfa(R_1[G \cdot H]) &\leq  \chi(G) \dfa(R_1[H]) + O(|V(G)|^2)\label{eq:intro dfa}\\
\alpha(G\cdot H) \leq \edp(R_2[G \cdot H]) &\leq \alpha(G)\edp(R_2 [H]) + O(|V(G)|^2)\label{eq:intro edp}
\end{align}
\danupon{TO DO: We should avoid using OPT}
\end{theorem}

See Section~\ref{sec:min-NFA} (especially, Corollary~\ref{cor: soundness product} and Lemma~\ref{lem:dfa-product}) and Section~\ref{sec:edp} (especially, Lemma~\ref{lem: edp product}) for the details and proofs of Eq.~\ref{eq:intro dfa} and Eq.~\ref{eq:intro edp}, respectively.
%
It only requires a systematic, simple calculation to show that these inequalities imply hardnesses of approximation; we formulate this implication as a ``meta theorem'' (see Section~\ref{sec:meta}) which roughly states that for large enough $k$, 
\begin{align}
f(R[G^k])\approx g(G^k)\label{eq:intro meta theorem}
\end{align}
where $G^k$ is $G*G*G*\ldots$ ($k$ times). (For an intuition, observe that  when $k$ is large enough, the term $\poly(|V(G)|)$ in Eq.~\ref{eq:intro gen form} will be negligible and an inductive argument can be used to show that  $g(G^k) \leq f(R[G^k])\leq g(G)^{O(k)}$ (recall that, in our case, $g$ is multiplicative)).
%
This means that the hardness of $f$\footnote{For conciseness, we will use $g$ and $f$ to refer to problems and their objective functions interchangeably.}
is at least the same as the hardness of $g$
on graph product instances $G^k$.
For the case of \DFA and \EDP, $R_1(G)$ and $R_2(G)$ increase the size of input size to $|V(G)|^2$ while $\alpha$ and $\chi$ have the hardness of $|V(G)|^{1-\epsilon}$. 
Thus, we get a hardness of $n^{1/2-\epsilon}$ where $n$ is the input size of \DFA and \EDP. This immediately implies a tight hardness for \EDP and an improved hardness of \DFA. How this translates to a hardness of $f$ depends on how much instance blowup the reduction $R[G^k]$ causes.
For our problems of \DFA and \EDP, it is a well known result that the hardness of $\alpha$ and $\chi$ stays roughly the same under the lexicographic product, i.e., $\alpha$ and $\chi$ on $G^k$ have a hardness of $|V(G^k)|^{1-\epsilon}$. The meta theorem and Theorem~\ref{thm:intro bounds of product} say that this hardness also holds for \DFA and \EDP. Since $R_1[G^k]$ and $R_2[G^k]$ increase the size of input instances by a quadratic factor --- from $|V(G^k)|$ to $n=|V(G^k)|^2$ --- we get a hardness of $n^{1/2-\epsilon}$ where $n$ is the input size of \DFA and \EDP. This immediately implies a tight hardness for \EDP and an improved hardness for \DFA. 

\subsection{Toward A Tight Hardness of the Minimum Consistent DFA Problem}
To get the tight $n^{1-\epsilon}$ hardness for \DFA, we have to adjust $R_1$ in Theorem~\ref{thm:intro bounds of product} to avoid the quadratic blowup. 
We will exploit the fact that, to get a result similar to Eq.~\ref{eq:intro meta theorem}, we only need a reduction $R$ defined on the $k$-fold graph product $G^k$ instead of on an arbitrary graph $G$ as in the case of $R_1$.
We modify reduction $R_1$ to $R_{1,k}$ that works only on an input graph in the form $G^k$ and produces an instance $R_{1,k}[G^k]$ of size almost linear in $|V(G^k)|$ while inequalities as in Theorem~\ref{thm:intro bounds of product} still hold, and obtain the following.

\begin{lemma}\label{thm:intro better DFA}
For any $k$, there is a reduction $R_{1,k}$ that reduces a graph $G^k=G\cdot G\cdot \ldots$ into an instance of the minimum consistent DFA problem of size $O(k\cdot |V(G^k)|\cdot |V(G)|^{2})$ such that
\begin{align}
\chi(G)^k\leq \dfa(R_{1,k}[G^k]) &\leq  \chi(G)^{2k}|V(G)|^4 \label{eq:intro dfa new}
\end{align}
\end{lemma}

The description of reduction $R_{1,k}$ and the proof of Theorem~\ref{thm:intro better DFA} can be found in Section~\ref{sec:min-DFA}.
Observe that the size $O(k\cdot |V(G^k)|\cdot |V(G)|^{2})$ of $R_{1,k}(G^k)$ is almost linear (almost $O(|V(G^k)|)$) as the extra  $O(k|V(G)|^{2})$ is negligible when $k$ is sufficiently large. 
Similarly, the term $|V(G)|^4$ in Eq.~\ref{eq:intro dfa new} is negligible and thus the value of $\dfa(R_{1,k}[G^k])$ is sandwiched by $\chi(G)^k$ and $\chi(G)^{2k}$. This means that if $\chi(G)$ is small (i.e., $\chi(G)\leq |V(G)|^{\epsilon}$), then $\dfa(R_{1,k}[G^k])$ will be small (i.e., $\dfa(R_{1,k}[G^k])\leq |V(G^k)|^{2\epsilon}$), and if $\chi(G)$ is large (i.e., $\chi(G)\geq |V(G)|^{1-\epsilon}$), then $\dfa(R_{1,k}[G^k])$ will be also large (i.e., $\dfa(R_{1,k}[G^k])\geq |V(G^k)|^{1-\epsilon}$). The hardness of $n^{1-\epsilon}$ for \DFA thus follows.

We note that in Theorem~\ref{thm:intro bounds of product}, we can replace \DFA by \NFA, a function corresponds to the minimum consistent NFA problem, thus getting a hardness of $n^{1/2-\epsilon}$ for this problem as well. 
This is, however, not yet tight. We would get a tight hardness if we can replace \DFA by \NFA in Theorem~\ref{thm:intro better DFA}, which is not the case. 
We also note that the proof for the tight hardness for the minimum consistent CNF problem follows from the same type of inequalities: We show that there exists a near-linear-size reduction $R_{3,k}$ from the minimum coloring problem to the minimum consistent CNF problem (with function $\cnf$) such that
\begin{align}
\chi(G)^k\leq \cnf(R_{3,k}[G^k])\leq \chi(G)^k |V(G)|^{O(1)}.\label{eq:intro cnf}
\end{align} 
%
The proofs of the bounds of graph products (Eq.~\ref{eq:intro dfa}, Eq.~\ref{eq:intro edp},Eq.~\ref{eq:intro dfa new} and Eq.~\ref{eq:intro cnf}) are fairly short and elementary; in fact, we believe that they can be given as textbook exercises. These proofs can be found in Section~\ref{sec:consistency-problem}, Section~\ref{sec:edp} and Section~\ref{sec:other}.


\subsection{Related Concept} 

Our pre-reduction graph product concept was inspired by the {\em graph product subadditivity} concept we previously introduced in \cite{ChalermsookLN-SODA13} (some of these ideas were later used in \cite{ChalermsookLN-FOCS13,ChalermsookLN-LATIN14}). 
There, we prove a hardness of approximation using the following framework. 
As before, let $f$ be an objective function of a problem that we intend to prove hardness and $g$ be an objective function of a problem whose hardness is already known.
%
%
We show that there are graph products $\oplus$, $*_e$, and $*$ such that  
\begin{itemize}
\item We can ``decompose'' $f(G*_e J)$: $g(G) \leq f(G*_e J) \leq g(G)+ f(G* J)$, and
\item $f((G\oplus H)* J)$ is ``subadditive'': $f((G\oplus H)* J) \leq f(G* J) + f(H* J)$.
\end{itemize}
We then use the above inequalities to show that if we let $G^k=G\oplus G\oplus \ldots $ ($k$ times), then
$$g(G^k)\leq f(G^k*_e J)\leq g(G^k)+ k f(G* J).$$
For large enough $k$, the term $kf(G* J)$ is negligible and thus $f(G^k*_eJ)\approx g(G^k)$. We use this fact to show that the approximation harness of $f$ is roughly the same as the hardness of $g$.
%
%
Observe that if we let $R[G]=G*_eJ$, the above inequalities can then be used to show that 
$$g(G\oplus H) \leq f(R[G\oplus H]) \leq g(G\oplus H)+f(R[G]) + f(R[H]).$$
In the problems considered in \cite{ChalermsookLN-SODA13}, one can easily bound $f(R[G])$ and $f(R[H])$ by $|V(G)|$ and $|V(H)|$, respectively. So, our meta theorem will imply that $f(G^k*J)\approx g(G^k)$, which leads to the approximation hardness of $f$. 
This means that the previous concept in \cite{ChalermsookLN-SODA13} can be viewed as a special case of our new concept where we restrict the reduction $R$ to be a graph product $R[G]=G*_eJ$. 
The way we use the reduction $R$ in this paper goes beyond this. For example, our reduction $R_2$ for \EDP as illustrated in Figure~\ref{fig:intro} cannot be viewed as a natural graph product. Moreover, our reduction $R_1$ reduces a graph $G$ to an instance of \DFA
which has {\em nothing} to do with graphs. (This is possible only when we abandon viewing reduction $R$ as a graph product.)
%
Our meta theorem also shows that bounds of graph products in a much more general form\danupon{Not sure if it's better to say ``loser form'' or ``loser bound''} can imply hardness results.\danupon{I just added the previous sentence.}
Finally, the way we exploit graph products using the reduction $R_{1,k}$ has never appeared in \cite{ChalermsookLN-SODA13}.

\danupon{TO DISCUSS: Should we say anything about being subadditive? The problem is that I don't know what to say.}


\subsection{Organization}

After giving necessary definitions in Section~\ref{sec:prelim}, we prove meta theorems in Section~\ref{sec:meta}. These theorems show that bounding $f(R[G*H])$ in a certain way will immediately imply a hardness result. They allow us to focus on proving appropriate bounds in later sections. In Section~\ref{sec:consistency-problem}, we prove such bounds for the consistency problems and their implications to the hardness of proper PAC-learning. In Section~\ref{sec:edp}, we prove such bounds of the edge-disjoint paths problem on DAGs. Bounds for other problems can be found in Section~\ref{sec:other}.

%
%


%% file: prelim.tex
\section{Preliminaries}\label{sec:prelim}

\subsection{Terms}
Given two graph $G$ and $H$, the {\em lexicographic product} of $G$ and $H$, denoted by $G\cdot H$, is defined as  
\begin{align*}
V(G\cdot H) &= V(G)\times V(H)  
              =\{(u, v) : u\in V(G) \mbox{ \underline{and} } v\in V(H)\}.\\
E(G\cdot H) &= \{(u, a)(v, b) : \mbox{$uv \in E(G)$ \underline{or} {\bf (}$u=v$ \underline{and} $ab\in E(H)${\bf )}}\}.
\end{align*}
Since the lexicographic product is the only graph product concerned in this paper, later on, we will simply use the term {\em graph product} to mean the lexicographic product.
We define the $k$-fold graph product of $G$, denoted by $G^k$, as
\begin{align*}
G^k = G\cdot G^{k-1}\mbox{ for any integer $k>1$}
\quad\mbox{and}\quad G^1 = G
\end{align*}
The properties of the lexicographic product that makes it becomes an import tools in proving hardness of approximation is that it multiplicatively increases the independent and chromatic numbers of graphs, without creating an overly dense resulting graph (the OR product also satisfies multiplicativity of independent and chromatic numbers, but it does not serve our purpose). 

\begin{theorem} \label{thm: property of lexicographic product}
Let $G$ and $H$ be any graphs. The followings hold on $G\cdot H$.
\begin{itemize}

\item $\alpha(G\cdot H)=\alpha(G)\alpha(H)$.

\item $\frac{\chi(G)\chi(H)}{\log|V(G)|}
       \leq
       \chi(G\cdot H)
       \leq 
       \chi(G)\cdot\chi(H)
      $.
\end{itemize}
In particular, for any $k\geq 1$, $\alpha(G^k)=\alpha(G)^k$ and $\chi(G)^{k-o(1)}\leq \chi(G^k) \leq \chi(G)^k$.
\end{theorem}


A {\em deterministic finite automaton} (DFA) is defined as a 5-tuple $(Q,\Sigma,\delta, q_0, F)$ where 
$Q$ is the set of states, 
$\Sigma$ is the set of alphabets,
$\delta: Q \times \Sigma \rightarrow Q$ is a transition function, 
$q_0$ is initial state, and 
$F\subseteq Q$ is the set of accepting states.
One can naturally extend the transition function $\delta$ into $\delta^*: Q \times \Sigma^* \rightarrow Q$ by inductively defining $\delta^*(q,x_1,\ldots, x_{\ell})$ as $\delta^*(\delta(q, x_1),x_2,\ldots, x_{\ell})$ and $\delta^*(q, null)=q$.  
We say that $M$ {\em accepts} $x$ if and only if $\delta^*(q_0, x) \in F$.  
The size of DFA $M$  is measured by the number of states of $M$, i.e., $|Q|$.
We say that a DFA is {\em acyclic} if there is no state $q \in Q$ and string $x$ such that $\delta^*(q,x) = q$.   
For NFA, the transition is defined by $\delta: Q \times \Sigma \rightarrow 2^Q$ instead, i.e., each transition possibly maps to several states.
An NFA $M$ accepts a string $x \in \Sigma^*$ if and only if the transition $\delta^*(q_0, x)$ contains an accepting state, i.e. $\delta^*(q_0, x) \cap F \neq \emptyset$. 

\subsection{Problems}
In this section, we list all problems considered in this paper. 
\paragraph{Minimum Consistency:}
In the {\sc Minimum Consistency} problem, denoted by {\sc MinCon($\hset$, $\fset$)}, we are given collections $\pset$ and $\nset$ of positive and negative sample strings in $\set{0,1}^*$, for which we are guaranteed that there is a hypothesis $h \in \hset$ that is consistent with all samples in $\pset \cup \nset$, i.e., $h(x) = 1$ for all $x \in \pset$ and $h(x)=0$ for all $x \in \nset$.   
Our goal is to output a function $f \in \fset$ that is consistent with all these samples, while minimizing $|f|$.  
In other words, $\hset$ and $\fset$ are the classes of the real hypothesis that we want to learn and those that our algorithm outputs respectively.  
This notion of learning allows our algorithm to output the hypothesis that is outside of the hypothesis class we want to learn. 

Now we need a slightly modified notion of approximation factor. 
For any instance $(\pset, \nset)$, we denote by $\opt_{\hset}(\pset, \nset)$ the size of the smallest hypothesis $h \in \hset$ consistent with $(\pset, \nset)$. 
Let $\aset$ be any algorithm for {\sc MinCon}($\hset$, $\fset$), i.e., $\aset$ always outputs the hypothesis in $\fset$. 
The approximation gauranteed provided by $\aset$ is: 
\[\sup_{\pset, \nset} \frac{|\aset(\pset, \nset)|}{\opt_{\hset}(\pset, \nset)}\]
 
With this terminology, the problem of learning DFA can be abbreviated as {\sc MinCon($DFA$, $DFA$)}.


\paragraph{Edge Disjoint Paths:}

In the {\em edge-disjoint paths} (EDP) problem, given a graph $G=(V,E)$
and a set of source-sink pairs $\{(s_1,t_1),\ldots,(s_k,t_k)\}$, our goal is to find a collection of paths 
$\pset=\{P_{i_1},P_{i_2},\ldots,P_{i_\ell}:i_j\in[k], \mbox{$P_{i_j}$ connects $s_{i_j}$ to $t_{i_j}$}\}$ 
that are edge disjoint while maximizing $|\pset|$. 
That is, we want to connects as many source-sink pairs as possible using a collection of edge-disjoint paths.

Our focus is on the special case of EDP where $G$ is a {\em directed acycle graph} (DAG).

\paragraph{Bounded-Length Edge-Disjoint Cycles:}
Given a graph $G= (V,E)$, the {\em cycle packing number} of $G$, denoted by $\nu(G)$, is the maximum integer $\ell$ such that there exist cycles $C_1,\ldots, C_{\ell}$ which are pairwise edge-disjoint in $G$. 
The {\em edge-disjoint cycle} problem (EDC) asks to compute the value of $\nu(G)$.
If we are additionally given an integer $k$, the $k$-cycle packing number of $G$, denoted by $\nu_k(G)$, is the maximum integer $\ell$ for which there exist pairwise edge-disjoint cycles $C_1,\ldots, C_{\ell}$ where each cycle $C_j$ contains at most $k$ vertices.  
In the {\em $k$-edge-disjoint cycle} problem ($k$-EDC), we are asked to compute $\nu_k(G)$ given an input $(G,k)$.  


\paragraph{Maximum Independent Set:} 
Given a graph $G=(V,E)$, a subset of vertices $S\subseteq V$ is {\em independent} in $G$ 
if and only if $G$ has no edge joining any two vertices in $S$. 
The {\em independence number} of $G$, denoted by $\alpha(G)$, is the size of a largest independent set in $G$. 
In the {\em maximum independent set} problem, we are asked to compute an independent set $S$ in $G$ with maximum size.

The following is the hardness results of the maximum independent set problem by H{\aa}stad\footnote{
The hardness results of the maximum independent set problem~\cite{Hastad96} and 
the graph coloring problem~\cite{FeigeK98} hold under the assumption $\mathrm{NP}\neq \mathrm{ZPP}$. 
The results were later derandomized by Zuckerman in \cite{Zuckerman07} 
and thus hold under the assumption $\mathrm{P}\neq\mathrm{NP}$.
\label{fn:derandomzied-indset-coloring}}, 
which will be used to obtain the hardness of EDP on DAGs.

\begin{theorem}[\cite{Hastad96}+\cite{Zuckerman07}]
\label{thm: ind set} 
Let $\epsilon >0$ be any constant.  
Given graph $G=(V,E)$, it is NP-hard to distinguish between the following two cases: 
\begin{itemize} 
\item (\yi:) $\alpha(G) \leq |V(G)|^{\epsilon}$ 

\item (\ni:) $\alpha(G) \geq |V(G)|^{1-\epsilon}$ 

\end{itemize}  
\end{theorem}

\paragraph{Chromatic Number:} 
Given a graph $G=(V,E)$, a {\em proper coloring} $\sigma: V(G) \rightarrow [c]$
is a function that assigns colors to vertices of $G$ so that any two adjacent vertices receive different colors assigned by $\sigma$
(i.e., $uv\in E\implies \sigma(u)\neq\sigma(v)$). 
The {\em chromatic number} of $G$, denoted by $\chi(G)$, is the minimum integer $c$ such that 
a proper coloring $\sigma: V(G) \rightarrow [c]$ exists, i.e., $G$ can be properly colored by $c$ colors. 
In the {\em graph coloring} problem, we are asked to compute a proper coloring $\sigma: V(G) \rightarrow [c]$ while minimizing $c$.
We will be using the following hardness of approximation result by Feige and Kilian~\cite{FeigeK98}\footref{fn:derandomzied-indset-coloring}.

\begin{theorem}[\cite{FeigeK98}+\cite{Zuckerman07}]
\label{thm: coloring} 
Let $\epsilon >0$ be any constant.  
Given graph $G=(V,E)$, it is NP-hard to distinguish between the following two cases: 
\begin{itemize} 
\item (\yi:) $\chi(G) \leq |V(G)|^{\epsilon}$ 

\item (\ni:) $\chi(G) \geq |V(G)|^{1-\epsilon}$ 

\end{itemize}  
\end{theorem} 

%% file: meta.tex
\section{Meta Theorems}\label{sec:meta}

In this section, we prove general theorems that will be used in proving most hardness results in this paper.
These theorems give {\em abstractions} of the (graph product) properties one needs to prove in order to obtain hardness of approximation results.  
Our techniques can be used to derive hardnesses for both minimization and maximization problems. 
For the former, the reduction is from minimum coloring, while the latter is obtained via a reduction from maximum independent set. 

Let us start with maximization problems. 
Suppose we have an optimization problem $\Pi$ such that any instance $I \in \Pi$ is associated with an optimal function $\opt_{\Pi}(I)$. 
We consider a transformation $R$ that maps any graph $G$ into an instance $R[G]$ of the problem $\Pi$. 
We say that a transformation $R$ satisfies a {\em low $\alpha$-projection property} with respect to a {\em maximization problem} $\Pi$ if and only if the following two conditions hold: 

\begin{itemize} 
\item (I) For any graph $G=(V,E)$, $\opt_{\Pi}(R[G]) \geq \alpha(G)$.   

\item (II) There are universal constants $c_1, c_2 > 0$ (independent of the choices of graphs) such that, for any two graphs $G$ and $H$, 
\[\opt_{\Pi}(R[G \cdot H]) \leq |V(G)|^{c_1} + \alpha(G)^{c_2} \opt_{\Pi}(R[H]). \]

\item (III) There is a universal constant $c_0 > 0$ such that 
\[\opt_{\Pi}(R[G]) \leq c_0 |R[G]|. \]

\end{itemize} 

Intuitively, the transformation $R$ with the low $\alpha$-projection property tells us that there are relationships between the optimal solution of the problem $\Pi$ on $R[G]$ and the independence number of $G$.
Instead of looking for a sophisticated construction of $R$, we focus on a ``simple'' transformation $R$ that establishes a connection on one side, i.e., $\opt_{\Pi}(R[G]) \geq \alpha(G)$, and the ``growth'' of $\opt_{\Pi}$ is ``slow'' with respect to graph products.   
Property (III) of the low $\alpha$-projection property says that the optimal is at most linear in the size of the instance, which is the case for almost every natural combinatorial optimization problem.

Next, we turn our focus to a minimization problem. 
In this case, we relate the optimal solution to the chromatic number of an input graph. 
Specifically, one can define the {\em low $\chi$-projection property} with respect to a {\em minimization problem} $\Pi$ as follows.

\begin{itemize} 
\item (I) For any graph $G=(V,E)$, $\opt_{\Pi}(R[G]) \geq \chi(G)$.   

\item (II) There are universal constants $c_1, c_2 > 0$ (independent of the choices of graphs) such that, for any two graphs $G$ and $H$, we have 
\[\opt_{\Pi}(R[G \cdot H]) \leq  |V(G)|^{c_1} + \chi(G)^{c_2} \opt_{\Pi}(R[H]).  \]

\item (III) There is a universal constant $c_0>0$ such that
\[\opt_{\Pi}(R[G]) \leq c_0 |R[G]|.\]

\end{itemize}

We observe that the existence of such reductions is sufficient for establishing hardness of approximation results, and the hardness factors achievable from the theorems depend on the size of the reduction.

\begin{theorem} [Meta-Theorem for Maximization Problems]
\label{thm: meta alpha}  
Let $\Pi$ be a maximization problem for which there is a reduction $R$ for $\Pi$ that satisfies low $\alpha$-projection property with $|R[G]| = O(|V(G)|^{d})$.  
Then for any $\epsilon >0$, given an instance $I$ of $\Pi$, it is NP-hard to distinguish between the following two cases: 

\begin{itemize} 
\item (\yi:) $\opt_{\Pi}(I) \geq |I|^{1/d-\epsilon}$ 

\item (\ni:) $\opt_{\Pi}(I) \leq |I|^{\epsilon}$
\end{itemize}  
\end{theorem}  

\begin{theorem} [Meta-Theorem for Minimization Problems]
\label{thm: meta chi}  
Let $\Pi$ be a minimization problem for which there is a reduction $R$ for $\Pi$ that satisfies low $\chi$-projection property with $|R[G]| = O(|V(G)|^{d})$, for some constant $d\geq 0$.
Then for any $\epsilon >0$, given an instance $I$ of $\Pi$, it is NP-hard to distinguish between the following two cases: 

\begin{itemize} 
\item (\yi:) $\opt_{\Pi}(I) \leq |I|^{\epsilon}$ 

\item (\ni:) $\opt_{\Pi}(I) \geq |I|^{1/d - \epsilon}$
\end{itemize}  
\end{theorem}

\subsection{Proof of Theorem~\ref{thm: meta alpha} (Meta Theorem for Maximization Problems)}
\label{sec:proof-meta-thm-max}

Consider a reduction $R$ that transforms a graph $G$ into an instance of $\Pi$ that satisfies the low $\alpha$-projection property. 
We analyze how the optimal value changes over $\ell$-fold lexicographic products. 

\begin{lemma} \label{lem:induction-meta-max}
For any positive integer $\ell$,  
$\opt_{\Pi}(R[G^{\ell}]) \leq \ell c_0 |V(G)|^{c_1+d+1} \alpha(G)^{2c_2 \ell}$  
\end{lemma}
\begin{proof} 
This is proved by induction on a positive integer $\ell$. 
The base case $\ell=1$ holds because $\opt_{\Pi}(R[G]) \leq c_0 |R[G]| \leq c_0 |V(G)|^{d}$. 
Assume that the induction hypothesis holds for any $\ell>1$, and consider $\opt_{\Pi}(R[G^{\ell+1}])$. 
By writing $G^{\ell+1} = G \cdot G^{\ell}$ and applying the low $\alpha$-projection property, we have 
\[
\opt_{\Pi}(R[G^{\ell+1}]) \leq |V(G)|^{c_1} + \alpha(G)^{c_2} \opt_{\Pi}(R[G^{\ell}])
\] 

Then, by applying induction hypothesis, we have
\begin{eqnarray*} 
\opt_{\Pi}(R[G^{\ell+1}]) & \leq & |V(G)|^{c_1} + \alpha(G)^{c_2} \left( \ell c_0 |V(G)|^{c_1 + d+1} \alpha(G)^{2 c_2 \ell} \right)\\ 
& \leq & |V(G)|^{c_1} +\alpha(G)^{c_2+2 c_2 \ell} \ell c_0 |V(G)|^{c_1 + d+ 1} \\
 & \leq& (\ell+1) c_0 |V(G)|^{c_1 +  d +1} \alpha(G)^{2 c_2 (\ell+1)}  
\end{eqnarray*} 
\end{proof}  

We note that the exponent of the term $\alpha(G)$ depends on $\ell$ (the number of times the product is applied), while that of $|V(G)|$ does not. 
Intuitively speaking, this is why the contribution of the term $|V(G)|^{c_1}$ vanishes after taking graph products.  

\paragraph{Hardness of Approximation.}

Now we prove the hardness of approximation result claimed in Theorem~\ref{thm: meta alpha}.
Start from graph $G$ as given by Theorem~\ref{thm: ind set}.
Then construct an instance $R[G^{\ell}]$ with $\ell = \ceil{1/\epsilon}$.
This results in the instance $R[G^{\ell}]$ of the problem $\Pi$ of size $N  = |R[G^{\ell}]| = O(|V(G)|^{\ell d})$. 

In the \yi, we have
\[
\opt_{\Pi}(R[G^{\ell}]) \geq \alpha(G^{\ell}) = \alpha(G)^{\ell} \geq  |V(G)|^{(1-\epsilon) \ell} = N^{1/d-O(\epsilon)}.
\]

In the \ni, we have
\[
\opt_{\Pi}(R[G^{\ell}]) \leq O(|V(G)|^{d+c_1 +1} \alpha(G)^{2 c_2 
\ell}).
\]
 
Since $\alpha(G) \leq |V(G)|^{\epsilon}$ in this case, we have 
\[
\alpha(G)^{2 c_2 \ell} \leq |V(G)|^{ 2 c_2} = |V(G)|^{O(1)} = N^{O(\epsilon)}.
\]
This implies that $\opt_{\Pi}(R[G^{\ell}]) \leq |V(G)|^{O(1)} N^{O(\epsilon)} = N^{O(\epsilon)}$, and the gap between \yi and \ni is $N^{1/d-O(\epsilon)}$.
This completes the proof.

\subsection{Proof of Theorem~\ref{thm: meta chi} (Meta Theorem for Minimization Problems)}
\label{sec:proof-meta-thm-min}

Similarly to the case of maximization problems, we can prove the following lemma by induction on integers $\ell$.
We shall skip the proof as it is the same as that of Lemma~\ref{lem:induction-meta-max} except that $\alpha(G)$ is replaced by $\chi(G)$.

\begin{lemma} \label{lem:induction-meta-min}
For any positive integer $\ell$, $\opt_{\Pi}(R[G^{\ell}]) \leq \ell c_0 |V(G)|^{c_1 + d+1} \chi(G)^{2 c_2 \ell}$. 
\end{lemma}


\paragraph{Hardness of Approximation}

Take the instance $R[G^{\ell}]$ with $\ell = \ceil{1/\epsilon}$. 

In the \yi, we have the following bound, which is slightly different from the case of the maximization problem.
\[
\opt_{\Pi}(R[G^{\ell}]) \geq \chi(G^{\ell}) \geq (\chi_f(G))^{\ell} \geq |V(G)|^{\ell(1-2\epsilon)} = N^{1/d- O(\epsilon)}.
\]
 
In the \ni, Lemma~\ref{lem:induction-meta-min} gives $\opt_{\Pi}(R[G^{\ell}]) \leq N^{O(\epsilon)}$. 
Thus, we have the desired gap, completing the proof.

\subsection{Overview of Applications} 

Most of the reductions in this paper are direct applications of the above two meta theorems. 
That is, we design the following reductions.

\begin{itemize} 
\item A reduction $R_{\EDP}$ for \EDP such that $|R_{edp}[G]| = O(|V(G)|^2)$ and satisfies $\alpha$-projection property. 
This implies a tight $n^{1/2-\epsilon}$ hardness of approximating \EDP on DAGs. 

\item A reduction $R_{fa}$ for MinCon such that $|R_{fa}[G]| = O( |V(G)|^2)$ and satisfies $\chi$-projection property. 
This gives $n^{1/2-\epsilon}$ hardness of approximating MinCon(NFA,ADFA). 

\end{itemize} 

Notice that the reduction $R_{fa}$ above is not tight. 
To obtain a tight result, we need $|R_{fa}[G]| = O(|V(G)|)$, and it seems difficult to obtain such a reduction.
We instead exploit the further structure of graph products and prove bounds of the form 
\[ \chi(G^k) \leq \opt(R'_{fa}[G^k]) \leq \chi(G)^{O(k)} |V(G)|^{O(1)}  \] 
Now our reduction size is smaller, i.e., $|R'_{fa}[G^k]| = |V(G)|^{(1+o(1))k}$ as opposed to $|R_{fa}[G^k]| = |V(G)|^{2k}$.
 Moreover, the reduction $R'_{fa}$ exploits the fact that the input graph is written as a $k$-fold product of graphs.
This more restricted form of graph products allows us to prove tight hardness (and PAC impossibility result) of DFA and DNF/CNF Minimization.

%% file: dfa-product.tex
\section{Hardnesses of Finite Automata Problems: Minimum Consistency and Proper PAC Learning}
\label{sec:consistency-problem}

We show in this section the hardness of the consistency problems for finite automata, as well as the implications on impossibility results for PAC learning.
We start our discussion by proving the hardness for 
{\sc MinCon}(ADFA,NFA), which includes the minimum consistent NFA problem ({\sc MinCon}(NFA,NFA)) as a special case.
Then we proceed to prove the tight hardness of approximating 
{\sc MinCon}(ADFA,DFA), which implies the tight hardness of approximating 
the minimum consistent DFA problem and also implies the impossibility result
for proper PAC-learning DFA. 

\subsection{Hardness of {\sc MinCon}(ADFA, NFA) via Graph Products}
\label{sec:min-NFA}

In this section, we show an $N^{1/2-\epsilon}$ hardness for {\sc MinCon}(ADFA,NFA).
Formally, we prove the following theorem. 

\begin{theorem} 
\label{thm: nfa hardness}
Let $\epsilon >0$ be any positive constant.  
Given two sets of positive and negative sample strings $\pset, \nset$ over alphabet $\Sigma= \set{0,1}$ with a total length of $N$ bits, it is NP-hard to distinguish the following two cases:
\begin{itemize} 
\item There is an acyclic deterministic finite automata of size $N^{\epsilon}$ that is consistent with all strings in $\pset \cup \nset$.
\item Any non-deterministic finite automata consistent with $\pset \cup \nset$ must have at least $N^{1/2 -\epsilon}$ states.    
\end{itemize}  
\end{theorem}  

This is done by designing a reduction $R[G]$ with $\chi$-projection property and $|R[G]| = O(|V(G)|^2)$.
Our proof in fact shows that the projection properties hold for both optimal DFA and NFA functions. 

\subsubsection{The Reduction $R$}


We will be working with binary strings, i.e., the alphabet set $\Sigma = \set{0,1}$.
Given a graph $G=(V,E)$, we construct two sets $\pset,\nset$ of positive and negative samples,
which encode vertices and edges of the graph.
We assume w.l.o.g. that $|V(G)| = 2^k$ for some integer $k$.
Therefore, each vertex $u \in V(G)$ can be associated with a $k$-bit string $\enc{u} \in \set{0,1}^k$. 

Now our reduction $R[G]$ is defined as follows. 
The positive samples are given by  
$$ \pset = \set{\enc{u}1\enc{u}^R: u \in V(G)}$$ 
and the negative samples are 
$$ \nset = \set{\enc{u}1 \enc{v}^R : uv \in E(G)}$$ 
We denote this instance of the consistency problem by an ordered pair $(\pset, \nset)$. 
Now we proceed to prove property (I), that any NFA consistent with $(\pset, \nset)$ must have at least $\chi(G)$ states.

\begin{lemma} \label{lem:independent-states}
Let $M=(Q, \Sigma, \delta, q_0, F)$ be an NFA that is consistent with $(\pset, \nset)$. 
Then for any vertex $u\in V(G)$, 
\[
\delta^*(q_0, \enc{u}) \not \subseteq \bigcup_{v: uv \in E(G)} \delta^*(q_0, \enc{v}).
\]
\end{lemma}  

\begin{proof}
Assume for contradiction that $\delta^*(q_0, \enc{u}) \subseteq \bigcup_{v: uv \in E(G)} \delta^*(q_0, \enc{v})$. 
Since $\enc{u}1\enc{u}^R$ is a positive sample, there is a state $q \in \delta^*(q_0, \enc{u})$
that leads to an accepting state (i.e., $\delta^*(q, 1 \enc{u}^R) \cap F \neq \emptyset$).
By the assumption, the state $q$ also belongs to another set $\delta^*(q_0, \enc{v})$ for some $v: uv \in E(G)$.

Now consider the string $\enc{v}1\enc{u}^R$, which is a negative sample because $vu\in E(G)$.
Since $q\in \delta^*(q_0, \enc{v})$ and $\delta^*(q, 1 \enc{u}^R) \cap F \neq \emptyset$,
the string $\enc{v}1\enc{u}^R$ must be accepted by $M$, a contradiction.
\end{proof} 

Lemma~\ref{lem:independent-states} implies in particular that, for each vertex $u \in V(G)$, the set 
$\delta^*(q_0, \enc{u}) \setminus \paren{\bigcup_{v: uv \in E(G)} \delta^*(q_0,\enc{v})}$ is not empty. 
Now denote by $\opt_{DFA}(R[G])$ and $\opt_{NFA}(R[G])$ the number of states in the minimum DFA and NFA that are consistent with the samples $R[G] = (\pset, \nset)_G$, respectively.

\begin{corollary}
\label{cor: soundness product}  
Any NFA $M$ that is consistent with $(\pset, \nset)_G$ must have at least $\chi(G)$ states.
Therefore, $\opt_{DFA}(R[G]) \geq \opt_{NFA}(R[G]) \geq \chi(G)$ for all $G$.  
\end{corollary} 

\begin{proof} 
For each state $q \in Q$, define a set $C_q = \set{u \in V(G): q \in \delta^*(q_0, \enc{u}) \setminus (\bigcup_{v: uv\in E(G)} \delta^*(q_0, \enc{v})}$. 
It is easy to see that $C_q$ is an independent set and thus form a proper color class of $G$.
Lemma~\ref{lem:independent-states} implies that each vertex $u \in V(G)$ belongs to at least one class.
So, $\set{C_q}_{q \in Q}$ gives a proper $|Q|$-coloring of $G$, 
implying that $|Q| \geq \chi(G)$. 
\end{proof}

  
\subsubsection{$\chi$-Projection Property}


We will consider a specific class of DFA $M= (Q, \Sigma, \delta, q_0, F)$, which we call {\em canonical DFA}.
Specifically, we say that a DFA is {\em canonical} if it has the following properties.  

\begin{itemize} 
\item The state diagram has exactly $\ell$ layers for some $\ell$, and each path from $q_0$ to any sink has length exactly $\ell$.

\item All accepting states are in the last layer.   

\end{itemize} 

Denote shortly by $\opt(R[G])$ the number of states in the minimum canonical DFA consistent with $R[G]$. 
So we have that $\opt(R[G]) \geq \opt_{DFA}(R[G]) \geq \opt_{NFA}(R[G])$. 
The following lemma gives the $\chi$-projection property for $\opt(\cdot)$ 

\begin{lemma} \label{lem:dfa-product}  
$\opt(R[G \cdot H]) \leq  \chi(G) (\opt(R[H]) + O(|V(G)|))$ 
\end{lemma}  

To prove this lemma, we show how to construct, given a canonical DFA for $R[H]$, a ``compact'' canonical DFA for $R[G \cdot H]$.
We note that one key idea here is to avoid exploiting the DFA for $R[G]$ but instead tries to use the color classes of $G$ in its optimal coloring to ``compress'' the DFA for $R[G \cdot H]$.  

\begin{proof}
Let $M_H= (Q_H, \set{0,1}, \delta_H, q_H, F_H)$ be the minimum DFA for the instance $R[H]$ whose number of states is $s=\opt(R[H])$ and has $\ell_H = 2 h +1$ layers for $h = \ceil{\log |V(H)|}$. 
Let $C_1,\ldots, C_B$ be the color classes of $G$ defined by the optimal coloring, so $B = \chi(G)$. Let $f: V(G) \rightarrow [B]$ be the corresponding coloring function.  
We will also be using several copies of a directed complete binary tree with $2^k$ leaves, 
where each leaf corresponds to a string in $\set{0,1}^k$ and is associated with a vertex in $V(G)$.
Call this directed binary tree $T_k$.  

We will use $M_H$ and $T_k$ to construct a new acyclic DFA $M$ that have at most $ B(s + O(|V(G)|))$ states and exactly $\ell = 2(k+\ell_H) +1$ layers. 
Now we proceed with the description of machine $M=(Q,\set{0,1}, \delta, q, F)$. 
We start by taking a copy of directed tree $T_k$, and call this copy $T_k^{(0)}$. 
The starting state $q$ is defined to be the root of $T_k^{(0)}$.
This is the {\em first phase} of the construction. Notice that there are $k$ layers in the first phase, so exactly $k$ positions of any input string will be read after this phase.
Each state in the last layer is indexed by $state(\enc{v})$ for each $v \in V(G)$. 

In the {\em second phase}, we take $B$ copies of the machines $M_H$ where the $j^{th}$ copy, 
denoted by $M_H^{(j)}=(Q_H^{(j)}, \{0,1\}, \delta_H^{(j)}, q_H^{(j)}, F_H^{(j)})$, 
is associated with color class $C_j$ defined earlier.  
For each vertex $v \in V(G)$, we connect the corresponding state $state(\enc{v})$ in the last layer of Phase 1 to the starting state $q_H^{(f(v))}$. 
This transition can be thought of as a ``null'' transition which can be removed afterward, but keeping it this way would make the analysis simpler.  
Since each copy of $M_H$ has $2\ell_H +1$ layers, now our construction has exactly $2\ell_H +k +1$ layers.  

In the final phase, we first extend all rejecting states in $M_H^{(j)}$ by a unified path until it reaches layer $2(\ell_H +k) +1$. 
This is a rejecting state $rej_0$.  
Now, for each $j=1,\ldots, B$, we connect each accepting state in the last layer of $M_H^{(j)}$ to the root in the copy $T^{(j)}_k$ again by a ``null'' transition, so we reach the desired number of layers now (notice that each root-to-leaf path has $2(k+\ell_H) +1$ states.)
The states in the last layer of $T^{(j)}_k$ are indexed by $state(j,\enc{v})$.
The accepting states of $M$ are defined as $F=\bigcup_{j=1}^B\set{state(j, \enc{u}^R): u\in C_j}$, and the rest of the states are defined as rejecting.  
This completes our construction.  
See Figure~\ref{fig:dfa-product:upperbound} for illustration.

\begin{figure}
\begin{center}  
 \includegraphics[trim = 0mm 0cm 0cm 0mm, scale=0.51]{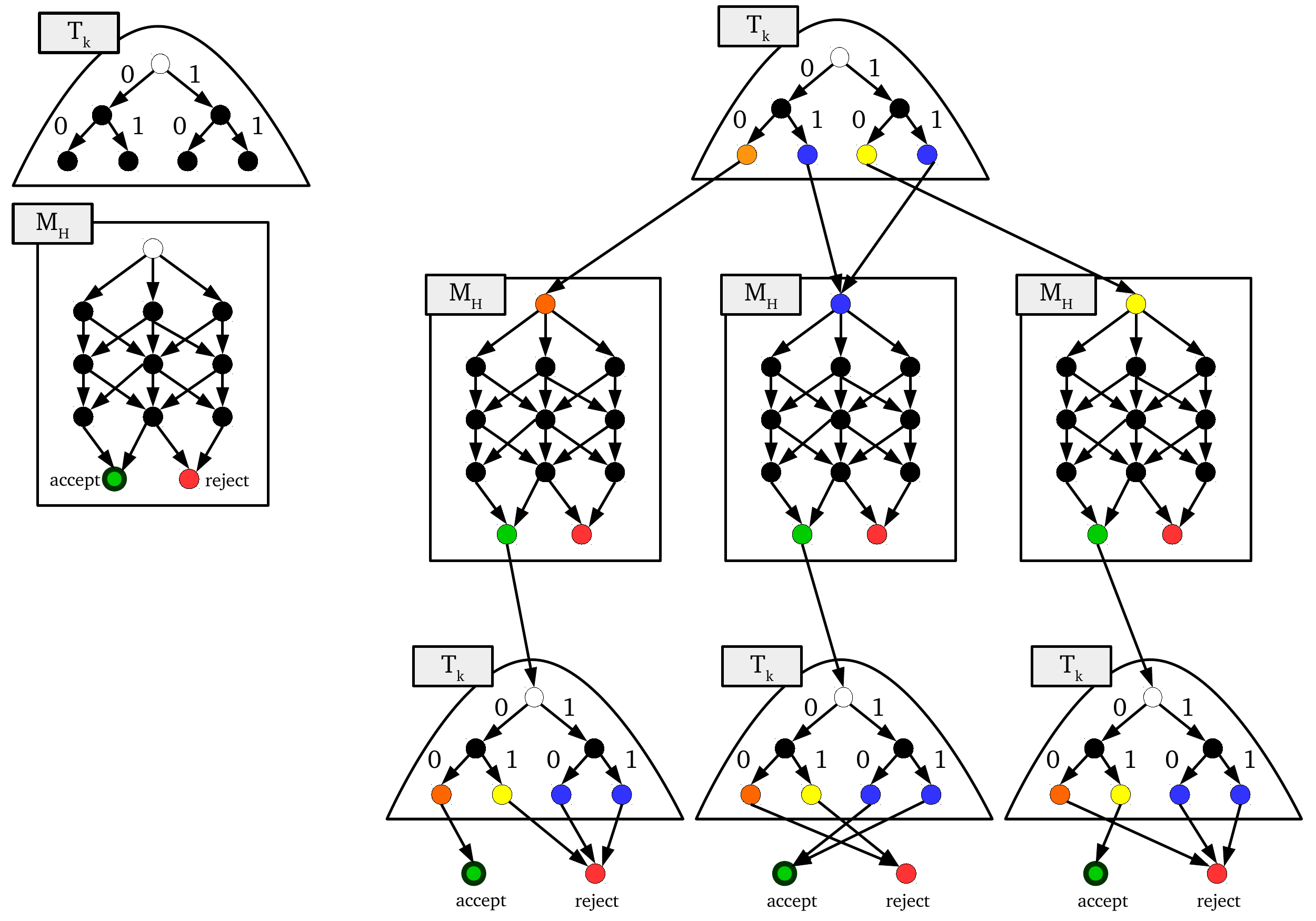}
\end{center}
\caption{The illustration of the construction in the proof of Lemma~\ref{lem:dfa-product}.
\label{fig:dfa-product:upperbound}
}
\end{figure}

The size of the construction is $|V(G)| + B s + O(|V(G)| B) = B(s + O(|V(G)|))$.  The next claim shows that the machine $M$ is consistent with samples obtained from the product of $G$ and $H$, which thus finish the proof.

\begin{claim} \label{claim:dfa-product:upperbound-consistent} 
Given a machine $M_H$ that is consistent with samples $R[H]$,
the machine $M$ constructed as above is consistent with samples $R[G \cdot H]$. 
\end{claim} 

\begin{proof} 
First we check the positive sample. 
For each vertex $(u,a) \in V(G \cdot H)$, the corresponding string $\enc{(u,a)} 1 \enc{(u,a)}^R$ can be thought of as $x=\enc{u} \enc{a} 1 \enc{a}^R \enc{u}^R$.
After the first $k$ transitions, the machine $M$ will stop at the state $q_H^{(f(u))}$.
Then the substring $\enc{a}1 \enc{a}^R$ will lead to an accepting state in $F_H^{(f(u))}$ (since $M_H$ is consistent with samples in $R[H]$).
Now, at the current state, we are at the root of the tree $T^{(f(u))}_k$, 
and we are left with the substring $\enc{u}^R$.
Since $u \in C_{f(u)}$, the substring $\enc{u}^R$ leads to an accepting state. 
This proves that the machine $M$ always accepts positive samples. 

Next, consider a negative sample $\enc{(u,a)} 1 \enc{(v,b)}^R$ generated by the edge $(u,a)(v,b) \in E(G)$.
Again, this can be thought of as $\enc{u} \enc{a} 1 \enc{b}^R \enc{v}^R$.
There are two possible cases: 

\begin{itemize} 
\item If $uv \in E(G)$, then the machine will enter $M_H^{(f(u))}$ after reading the substring $\enc{u}$.
Next, the machine reads the substring $\enc{a} 1 \enc{b}^R$.
If it manages to reach the third phase without rejection 
(i.e., $M_H$ accepts $\enc{a} 1 \enc{b}^R$),
then it will enter the tree $T_k^{(f(u))}$.
Note that there is no edge joining two vertices in $C_{f(u)}$
because it is a color class.
Thus, the substring $\enc{v}$ leads to a rejection because $uv\in E(G)$ implies that $v \not \in C_{f(u)}$.

(Notice that the rejection does not depend on what happens inside $M_H$.)  

\item If $u = v$ and $ab \in E(G)$, then after the first $k$ transitions, the machine enters $M_H^{(f(u))}$ with the input string $\enc{a} 1 \enc{b}^R$ for $ab \in E(H)$.
Since $M_H$ is consistent with samples in $R[H]$, this would lead to a rejection in $M_H^{(f(u))}$ and therefore in $M$.  

\end{itemize}   

\end{proof} 
 
\end{proof}

We will also need the base case condition as required by the low $\alpha$-projection property.  

\begin{lemma}
\label{lem: base}  
For any graph $G$, $\opt(R[G]) \leq O(|V(G)|^2 \log |V(G)|)$ 
\end{lemma} 
\begin{proof} 
We simply use the tree $T_{2k+1}$ with the initial state $q$ at the root of $T_{2k+1}$, where each vertex at the leaf can be associated with a string in $\set{0,1}^{2k+1}$. 
We simply define the accepting states to be those that correspond to the strings of the form $\enc{u} 1 \enc{u}^R$. 
The size of the construction is $2^{2k+1}  = O(|V(G)|^2 \log |V(G)|)$.  
\end{proof}

%% file: dfa-tight.tex
\subsection{Tight Hardness for {\sc MinCon}(ADFA, DFA)}
\label{sec:min-DFA}

Notice that the construction in the previous section is not tight because the size of the negative samples in $R[G]$ is large compared to the number of vertices in graph $G$, i.e., $|\nset| = \Theta(|V(G)|^2)$. 
To handle this problem, we take into account the structure of the lexicographic product and ``encode'' negative samples in a more compact form, i.e. we ideally want the construction size to be nearly linear on $|V(G)|$, i.e., $|\nset| = O(|V(G)|^{1+o(1)})$, instead of quadratic.

To this end, we construct a reduction $R_k[G^k]$  
We remark that, while the construction in this section gives tighter results for DFA, ADFA, and OBDD, it does not apply to NFA.  

\subsubsection{The Reduction $R_k[G^k]$}
\label{sec:min-DFA:reduction}

We show a reduction $R_k[G^k]$ of size $|R_k[G^k]| = |V(G)|^{k(1+o(1))}$. 
Consider a graph $H= G^k$ (the $k$-fold lexicographic product of $G$). 
We will encode the edge structures of $H$ into the positive and negative samples as follows. 

{\bf Positive Samples:}
For each $\vec{u}= (u_1,\ldots, u_k) \in V(H)$, define a positive sample 
\[
pos(\vec{u}, i) = \enc{u_1}\ldots \enc{u_k} 1 \enc{u_1} \ldots \enc{u_i}.
\]

The set of all positive samples is denoted by 
\[\pset = \set{pos(\vec{u}, i): u \in V(H), i =1\ldots k  } \] 

{\bf Negative Samples:}
For each a pair of vertices $\vec{u} \in V(H)$ and $v \in V(G)$ such that $u_i v \in E(G)$, define a negative sample
\[
neg(\vec{u}, v, i) = \enc{u_1}\ldots \enc{u_k} 1 \enc{u_1} \ldots \enc{u_{i-1}} \enc{v_i}\]

The set of all negative samples is denoted by 
\[
\nset = \set{neg(\vec{u}, v, i): \vec{u} \in V(H), v\in V(G), u_i v \in E(G), i = 1,\ldots, k } \]   

Intuitively, an edge in the the input graph represents a {\em conflict} between two vertices.
Negative samples are thus defined to capture a conflict (an edge) in the product of graphs between vertices $\vec{u}$ and $\vec{v}$ at coordinate $i$.
Notice that the size of positive and negative samples are $|\pset| = k n^k$ and $|\nset| = k n^k |E(G)| \leq k n^{k+2}$.

Let $\opt_{ADFA}(\cdot)$ denote the number of states in the optimal acyclic DFA that is consistent with the samples.  
We will prove the following lemma.

\begin{lemma}
\label{lmm:ineq-tight-min-DFA} 
$\chi(H) \leq \opt_{DFA}(\pset, \nset) \leq \opt_{ADFA}(\pset, \nset) \leq \chi(G)^{2k} |V(G)|^4.$
\end{lemma}

The bound $\opt_{DFA}(\pset, \nset) \leq \opt_{ADFA}(\pset, \nset)$ is trivial.
For the other bounds, we will prove the left and right-hand side inequalities of Lemma~\ref{lmm:ineq-tight-min-DFA} in Section~\ref{sec:lower-bound-tight-DFA} and Section~\ref{sec:upper-bound-tight-DFA}, respectively.
The hardness result then follows trivially from Theorem~\ref{thm: coloring} and Theorem~\ref{thm: property of lexicographic product}. 
In particular, taking the hard instance of the graph coloring problem as in Theorem~\ref{thm: coloring}, we have that

\medskip

\yi: $\opt_{DFA}(\pset, \nset) \leq \chi(G)^{2k}n^4 \leq n^{O(1)}$.

\medskip

\ni: $\opt_{DFA}(\pset, \nset) \geq \chi(H) \geq \chi(G)^{k-o(1)} \geq n^{(k-O(1))}$.

\medskip

Since $|\pset|+|\nset|=O(kn^{k+2})$, this implies the hardness gap of $(|\pset|+|\nset|)^{1-\varepsilon}$, for any $\varepsilon>0$.  

\subsubsection{The Lower Bound of $\opt_{DFA}$}
\label{sec:lower-bound-tight-DFA}

First, we show the lower bound for $\opt_{DFA} (\pset, \nset)$.
Let $M=(Q,\Sigma, \delta, q_0, F)$ be a DFA consistent with $(\pset, \nset)$. 
We construct from $M$ a $|Q|$-coloring of $H$:
For each state $q \in Q$, we define a color class $C_q = \set{\vec{u}: \delta^*(q_0, \enc{\vec{u}}) = q}$.
Since $M$ is deterministic, each vertex must get at least one color.

\begin{lemma} 
For any vertices $\vec{u}, \vec{v} \in C_q$, $\vec{u} \vec{v} \not \in E(H)$.
That is, $C_q$ is a proper color class of $H$.  
\end{lemma}    
\begin{proof} 
Suppose to a contrary that there is a pair of vertices$\vec{u}, \vec{v} \in C_q$ such that $\vec{u} \vec{v} \in E(H)$.
Since $H$ is obtained by the lexicographic product, there exists a coordinate $i$ in which $\vec{u}$ and $\vec{v}$ conflict, 
i.e., $u_j = v_j$ for all $j < i$ and $u_i v_i \in E(G)$.
%
We know that $\delta^*(q,1 \enc{u_1} \ldots \enc{u_i}) \in F$ because 
$pos(\vec{u}, i)=\enc{\vec{u}} 1 \enc{u_1}\ldots \enc{u_i}$ is a positive sample.
Since $\delta^*(q_0,\enc{\vec{u}})=\delta^*(q_0,\enc{\vec{v}})=q$,
we must also have $\delta^*(q_0,\enc{v}1\enc{u_1} \ldots \enc{u_i})=\delta^*(q,1\enc{u_1} \ldots \enc{u_i}) \in F$.
But, this contradicts the fact that 
$neg(\vec{v}, \vec{u}, i)= \enc{\vec{v}} 1 \enc{v_1} \ldots \enc{v_{i-1}} \enc{u_i} = \enc{\vec{v}} 1 \enc{u_1} \ldots \enc{u_i}$ is a negative sample.
\end{proof}

\subsubsection{The Upper bound of $\opt_{ADFA}$}
\label{sec:upper-bound-tight-DFA}

Now we need to argue that there is an acyclic DFA $M$ of size $\chi(G)^{2k} n^4$.   
Suppose $V(G) = \set{0,1}^\ell$. 
Let $c=  \chi(G)$, and $\sigma: V(G) \rightarrow [c]$ be an optimal coloring of $G$. 
Our construction has two steps.
First, we construct a complete rooted $c$-ary tree with $2k$ level, namely $S$.
Note that $S$ is a directed tree whose edges are oriented toward leaves.
Each vertex in $S$ except the root is associated with one color class from $\sigma$.
In particular, for each internal vertex $a$ of $S$,
each child $x$ of $a$ is associated with a distinct color from $[c]$.
We define the coloring of $S$ by $\rho: V(S) \rightarrow [c]$.
Second, we replace each vertex $a$ of $S$ by a complete binary tree $T$ with $n$ leaves;
we denote this copy of $T$ by $T_a$. 
Each leaf $q$ of $T_a$ is associated with a vertex $u$ of $G$ and thus has a color $\sigma(u)$ assigned.
(We abuse $\sigma(q)=\sigma(u)$ to mean a color of $q$.)
For any vertex $x$ in $S$ that is a child of $a$,
we join every leaf $q$ of $T_a$ with color $\sigma(q)=\rho(a)$ to the root $r$ of $T_x$.
The transition edge $qr$ is a null transition unless $a$ is a vertex at level $k$ in $S$; for the case that $a$ is at level $k$, the transition edge $qr$ is labeled ``1''.
(Note that a null transition edge $qr$ means that we will merge $q$ and $r$ in the final construction. It is easy to see that this results in a DFA (not NFA) because $S$ is a tree.)
It can be seen that the constructed directed graph 
has a single source vertex (i.e., a vertex with no incoming edges), 
which we define as a starting state $q_0$.

To finish the construction, we define accepting states.
Let $a_0$ be the root of $S$.
Consider a vertex $a_i$ at level $i>k$ in $S$ and its corresponding tree $T_{a_i}$.
Since $S$ is a tree, there is a unique path from $a_1$ to $a_i$, namely,
$P=a_0,\ldots,a_k,a_{k+1}\ldots,a_i$. 
For each leaf $q$ of $T_{a_i}$, we define $q$ as an accepting state if and only if $\sigma(q)=\rho(a_{i-k+1})$, i.e., $q$ and $a_{i-k+1}$ receive the same color.
See Figure~\ref{fig:dfa-tight:upperbound} for illustration.

Each copy of $T$ has at most $2n$ vertices,
and $S$ has at most $c^{2k}$ vertices.
Thus, the size of the DFA $M$ is at most $2n c^{2k}$.
Also, observe that $M$ is acyclic.

\begin{figure}

\begin{center}  
\includegraphics[trim = 0mm 0cm 0cm 0mm, scale=0.51]{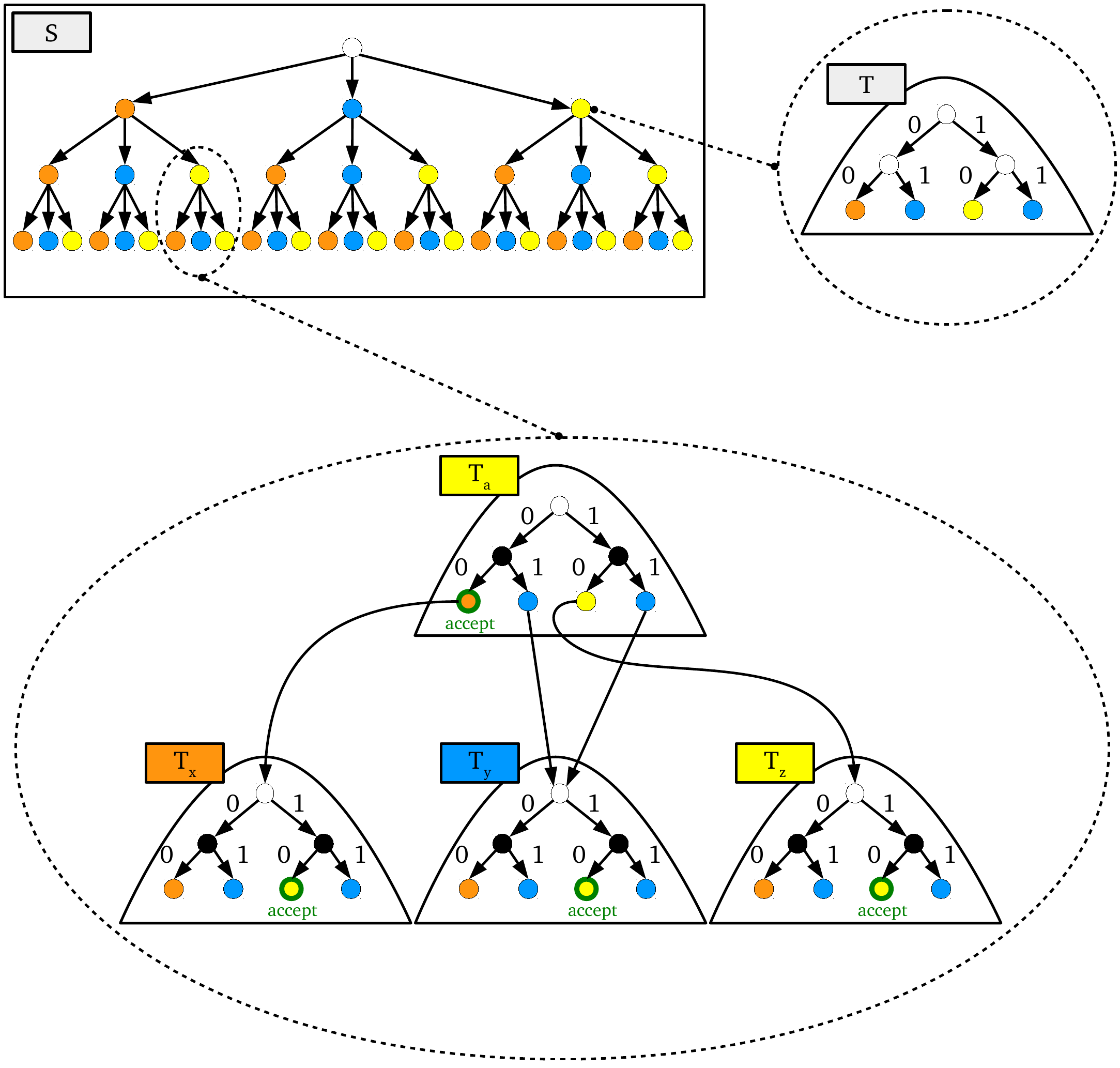}
\end{center}
\caption{The illustration of the construction in the proof of Lemma~\ref{lmm:ineq-tight-min-DFA}.
\label{fig:dfa-tight:upperbound}
}
\end{figure}

\begin{lemma} 
The DFA $M$ is consistent with $(\pset, \nset)$. 
\end{lemma}
\begin{proof} 

Consider any sample $\vec{u}\in\pset\cup\nset$,
which must be of the form:
\[
\vec{u}=\enc{u_1}\ldots\enc{u_k}1\enc{u_{k+1}}\ldots\enc{u_{k+i}}
\quad\mbox{for $i:1\leq i \leq k$
 and $u_j\in V(S)$ for all $j=1,\ldots,k+i$}. 
\]

Note that $u_j=u_{k+j}$ for all $j=1,2,\ldots,i-1$. 
The transition $\delta^*(q_0,\vec{u})$ forms a path $P$ in $M$,
which traverses from the starting state $q_0$ to some state $q_j$.
(That is, $q_j=\delta^*(q_0,\vec{u})$.)
By construction, $P$ corresponds to the path $a_1,\ldots,a_i$ in $S$
and thus must visit a leaf $q_j$ of tree $T_{a_j}$, for $j=1,\ldots,i$.
Moreover, each $q_j$ is associated with vertex $u_j\in V(G)$. 
Notice that $\rho(a_j)=\sigma(u_{j-1})$ because we have an edge from $q_{j-1}$ to $T_{a_j}$ if and only if $q_{j-1}$ and $a_j$ receive the same color.

If $\vec{u}$ is a positive sample in $\pset$, then we have $u_{i-k}=u_i$.
(Note that $q_j$ and $u_j$ receive the same color for all $j=1,2,\ldots,k$.)
Since $a_{i-k+1}$ has the same color as $q_{i-k}$ (and so does $u_{i-k}=u_i$), we have $\rho(a_{i-k+1})=\sigma(q_i)$.
Thus, $q_i$ is an accepting state.

If $\vec{u}$ is a negative sample in $\nset$, then we must have an edge $u_{i-k}u_i\in E(G)$. 
So, $u_{i-k}$ and $u_i$ receive different colors.
Since $\rho(a_{i-k+1})=\sigma(u_{i-k})$, 
it follows that $\rho(a_{i-k+1})\neq\sigma(q_i)$.
Thus, $q_i$ is not an accepting state.
This proves that $M$ is consistent with both positive and negative samples.
\end{proof}

%% file: pac-learning.tex
\subsection{Hardness of Proper PAC-Learning}
\label{sec:hardn-pac-learning}

Here we show that DFAs are not PAC-learnable. That is, we prove Corollary~\ref{cor:pac-learning-dfa}.
We will use the connection between PAC learning and the existence of an {\em Occam algorithm}, defined as follows. 

\begin{definition} 
An Occam algorithm for a hypothesis class $\hset$ in terms of function classes $\fset$ is an algorithm $\aset$ that for some constant $k \geq 0$ and $\alpha <1$, the following guarantee holds.
Let $h \in \hset$ has size $n$ and represents some language $L(h)$. Then on any input of $s$ samples of $L(r)$, each of length at most $m$, the algorithm $\aset$ outputs an element $h \in \hset$ of size at most $n^k m^k s^{\alpha}$ that is consistent with each of the $s$ samples.  
\end{definition} 

Therefore, an Occam algorithm for DFA is the case when $\hset = {\sf DFA}=  \fset$, and the measure of the size of each hypothesis $h \in {\sf DFA}$ is the number of states.  
It is known that PAC learnability of DFA implies the existence of an Occam algorithm for the same hypothesis class as stated formally in the following theorem.

\begin{theorem}[\cite{BoardP92}, statement from~\cite{Pitt89survey}]\label{thm:DFA-PAC-larning-implies-Occam}
If DFAs are properly PAC-learnable, then there exists a randomized Occam algorithm for DFA that runs in polynomial time.  
\end{theorem} 

Theorem~\ref{thm:DFA-PAC-larning-implies-Occam} implies that, to prove Corollary~\ref{cor:pac-learning-dfa}, it suffices to rule out the existence of a randomized Occam algorithm for DFA, which is shown in the next Theorem. 

\begin{theorem} \label{thm:no-occam-dfa}
Unless $\mathrm{NP} = \mathrm{RP}$, there is no polynomial time randomized Occam algorithm for DFA. 
\end{theorem}

\begin{proof} 
We prove by contrapositive.
Assume that there is a randomized Occam algorithm $\aset$ for DFA with parameters $(k, \alpha)$ for some constants $k\geq 0$ and $0\leq \alpha<1$.
Then we argue that there would exist and an algorithm that distinguishes between the \yi and \ni given in Theorem~\ref{thm:hardness-mincon-dfa}.
To see this, take an instance of MinCon(ADFA,DFA) as in Theorem~\ref{thm:hardness-mincon-dfa}.
So, we have a pair of sets $(P,N)$ of $N$ samples, each of length $O(\log N)$. 
The parameters of the Occam algorithm $\aset$ are thus $s=N$ and $m=O(\log N)$. 

We choose the parameter $\epsilon$ in Theorem~\ref{thm:hardness-mincon-dfa} to be $\epsilon = (1-\alpha)/(2k+1)$.
 
In the \yi, there is a DFA of size $N^{\epsilon}$ consistent with the samples.
Thus, our hypothesis class has size $n=N^{\epsilon}$. 
By definition, the Occam algorithm $\aset$ gives us a DFA $M$ of size
\[
|M| \leq N^{\epsilon k} \cdot (\log N)^k \cdot N^{\alpha} 
    \leq N^{\alpha + (2 k) \epsilon}
    \leq N^{\alpha + (2 k) \frac{1-\alpha}{2k+1}}
    = N^{1-\frac{1-\alpha}{2k+1}}
    = N^{1-\epsilon}
\]

In the \ni, any DFA $M$ consistent with $(P,N)$ has size $|M| > N^{1-\epsilon}$.

Therefore, the randomized Occam algorithm $\aset$ can distinguish between the \yi and \ni in Theorem~\ref{thm:hardness-mincon-dfa}, implying that $\mathrm{NP}=\mathrm{RP}$.
This completes the proof.
\end{proof}  

A similar but weaker theorem can be proven for the case of NFAs. Indeed, we rule out the existence Occam algorithm for NFA with parameter $0\leq \alpha \leq 1/2$, assuming that $\mathrm{NP}\neq \mathrm{RP}$.

\begin{theorem} \label{thm:no-occam-nfa}
Unless $\mathrm{NP} = \mathrm{RP}$, there is no polynomial time randomized Occam algorithm for NFA with parameter $0\leq \alpha \leq 1/2$.  
\end{theorem}

\begin{proof} 
The proof is essentially the same as that of Theorem~\ref{thm:no-occam-dfa} with slightly different parameters. 

We prove by contrapositive.
Assume that there is a randomized Occam algorithm $\aset$ for NFA with parameters $(k, \alpha)$ for some constants $k\geq 0$ and $0\leq \alpha \leq 1/2$.
We will show that the algorithm $\aset$ can be used to distinguishes between the \yi and \ni given in Theorem~\ref{thm:hardness-mincon-nfa} and thus implying that $NP=RP$. 

Take an instance of MinCon(ADFA,NFA) as in Theorem~\ref{thm:hardness-mincon-dfa}.
So, we have a pair of sets $(P,N)$ of $N$ samples, each of length $O(\log N)$. 
The parameters of the Occam algorithm $\aset$ are thus $s=N$ and $m=O(\log N)$. 

We choose the parameter $\epsilon$ in Theorem~\ref{thm:hardness-mincon-dfa} to be $\epsilon = (1/2 - \alpha)/(2k+1)$.
 
In the \yi, there is an NFA of size $N^{\epsilon}$ consistent with the samples.
Thus, our hypothesis class has size $n=N^{\epsilon}$. 
By definition, the Occam algorithm $\aset$ gives us a NFA $M$ with size
\[
|M| \leq N^{\epsilon k} \cdot (\log N)^k \cdot N^{\alpha} 
    \leq N^{\alpha + (2 k) \epsilon}
    = N^{\alpha + (2 k) \frac{1/2 - \alpha}{2k+1}}
    = N^{1/2-\frac{1/2-\alpha}{2k+1}}
    = N^{1/2-\epsilon}
\]

In the \ni, any NFA $M$ consistent with $(P,N)$ has size
$|M| > N^{1/2-\epsilon}$.

Therefore, the randomized Occam algorithm $\aset$ can distinguish between the \yi and \ni in Theorem~\ref{thm:hardness-mincon-nfa}, implying that $\mathrm{NP}=\mathrm{RP}$.
This completes the proof.
\end{proof}  

\begin{corollary} 
Unless $NP= RP$, there are no Occam algorithms for the following hypothesis classes: 

\begin{itemize} 
\item Deterministic Finite Automata (DFA) 


\item Acyclic Deterministic Finite Automata (ADFA) 

\item Ordered Branching Decision Diagram (OBDD) 
\end{itemize}
In particular, for any $\epsilon \in (0,1), k >0$, the minimum consistent hypothesis problems for these classes are $N^{1-\epsilon} \opt^{k}$-hard to approximate unless $NP = RP$.  
\end{corollary}

%% file: edp-product.tex
\section{Hardness of EDP on DAGs}
\label{sec:edp}

In this section, we prove the $|V(G)|^{1/2-\epsilon}$ hardness of approximating EDP on DAGs and packing vertex-disjoint bounded length cycles.  
We will first show the construction for EDP, and later we argue that a slight modification of the construction yields the hardness of packing vertex-disjoint bounded length cycles. 

\subsection{Reduction $R$}
\label{sec:edp:simplex-proof}

We first define the canonical reduction $\wall[G]$ formally. 
Given a graph $G=(V,E)$ on $n$ vertices, the switching graph of $G$, denoted by $\wall [G]$, is a graph defined on a plane and constructed in two steps as follows. 
The coordinates of graph $\wall[G]$ lie in the box formed by the corners $(0,0)$ and $(n,n)$.

\medskip

{\sc First Step:} For each vertex $i \in V(G)$, we draw a line segment $\ell_i$ on the plane connecting vertices $s_i$ and $t_i$ as shown in Figure~\ref{fig: transformation}. 
To be precise, the line $\ell_i$ goes from the coordinate $(n+1-i,0)$ to the coordinate $(n+1-i,i)$ of the grid and then goes to the coordinate $(0,i)$.
For each pair of vertices $i,j \in V(G)$, we have an {\em intersection point $y_{i,j}$} at the crossing point of lines $\ell_i$ and $\ell_j$.
Some of these intersection points will be later defined as vertices in the switching graphs whereas others are just a crossing points in the plane embedding.
We call this graph $\awall [G]$ which will also be crucial in our analysis.  
Edges in $\awall [G]$ are directed from left to right and top to bottom.  

\medskip

{\sc Second Step:} For each edge $ij \in E(G)$, we split $y_{i,j}$ into two vertices $x^{\mathrm{in}}_{i,j}$ and $x^{\mathrm{out}}_{i,j}$ and have a directed edge $e_{i,j} = x^{\mathrm{in}}_{i,j} x^{\mathrm{out}}_{i,j}$ in the graph $\wall[G]$.
Otherwise, if $ij \not\in E(G)$, the intersection point $y_{i,j}$ is replaced by an uncrossing as in Figure~\ref{fig: transformation}.

\begin{figure}[hbt]
\begin{center} 
\includegraphics[scale=0.5, clip=true, trim=  0 3cm 4cm 0cm] {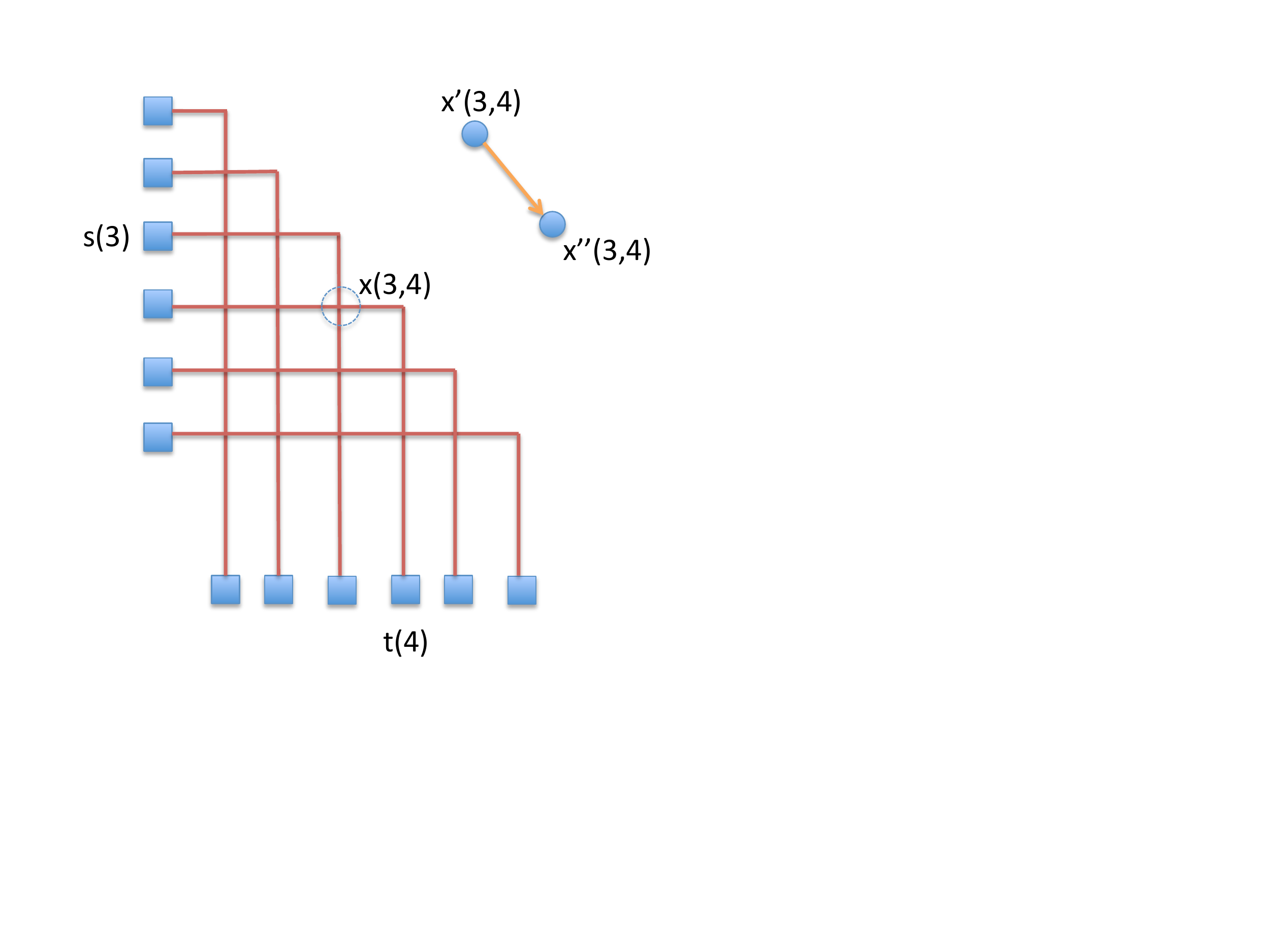}
\caption{The graph $\protect\wall[G]$ where $(3,4) \in E(G)$ but $(4,5) \not\in E(G)$}
\label{fig: transformation}
\end{center} 
\end{figure}
\danupon{I removed ``TO DO: Draw a new figure'' from the caption.}



%
%
%

First, the following lemma establishes a (simple) connection between \EDP and the maximum independent set problem.

\begin{lemma}
For any graph $H$, $\edp(\wall [H]) \geq \alpha(H)$. 
\end{lemma} 

\begin{proof} 
Let $S \subseteq V(H)$ be any independent set in $H$. 
We define the collection of paths $\pset_S = \set{P_i}_{i \in S}$ in graph $\wall [H]$. 
Since $S$ is an independent set, any pair of paths $P_i$ and $P_j$ for $i,j \in S$ are disjoint by construction.  
\end{proof} 

Unfortunately, the converse of this inequality does not hold within any reasonably small factor. 
In fact, there is a graph $H$ for which $\alpha(H) = 2$ but $\edp(\wall [H]) = \Omega(n)$; see Appendix~\ref{sec: bad}.
Therefore, we focus on proving the low $\alpha$-projection property.

\subsection{$\alpha$-Projection Property}

For technical reasons, we will need to analyze a slightly different measure from the optimal value $\edp(\wall [G])$. 
This notion will be a weaker notion of feasible solutions for EDP. 
We say that a collection of disjoint paths $\pset= \set{P_1,\ldots, P_{\ell}}$ is {\em orderly feasible} if for any pair $P=(s_i,\ldots, t_{j})$ and $P' =(s_{i'}, \ldots, t_{j'})$ such that $i < i'$, then it must be the case that $j < j'$; for instance, in an orderly feasible set, if we connect $s_1$ to $t_3$, it must be the case that $s_2$ is connected to $t_j$ for $j >3$.
Intuitively, in an orderly feasible set $\pset$, a path is allowed to start from $s_i$ and ends at some sink $t_j$ for $j \neq i$, but every pair of paths in $\pset$ is forced to ``cross'' at some point.    
Observe that any collection of feasible edge disjoint paths must also be orderly feasible.  
As a consequence, if we define $\tildeEDP (\wall [G])$ as the maximum cardinality of all orderly feasible collections of paths, then we have that $\tildeEDP (\wall [G]) \geq \edp (\wall [G])$.  


The following observation is more or less obvious. 

\begin{observation} 
For any graph $G$, $\tildeEDP(\wall [G]) \leq |\wall [G]| $
\end{observation}

Next, the following lemma will finish the proof of the low $\alpha$-projection property.

\begin{lemma}
\label{lem: edp product}  
For any two graphs $G$ and $H$, 
\[\tildeEDP(\wall[G \cdot H]) \leq 3|V(G)|^2 + \alpha(G)\tildeEDP(\wall [H]) \]
\end{lemma}  

We will spend the rest of this section to prove the lemma.  









\subsection{Geometry of Paths: Regions, switching boxes, and configurations}

This section discusses the structure of the graph $\wall [G \cdot H]$ and a feasible solution for EDP in $\wall [G \cdot H]$.
We define some terminologies that will be needed in the analysis.  

\paragraph{Ordering of Paths} 
We need a notion of ``ordering'' of edge-disjoint paths with respect to certain curve.
We think of graph $\awall [G]$ as being drawn on the plane with standard $x$ and $y$ coordinates. 
All sources and sinks are on $y$ and $x$ axes respectively.  
 
For any collection of edge-disjoint paths $\pset$ in $\wall [G]$, one can naturally map these paths on the graph $\awall [G]$ and think of them as curves on the plane. 
A continuous curve $C: [0,1] \rightarrow {\mathbb R}^2$ is said to be {\em good} if for all $t < t'$, point $C(t)$ is dominated by point $C(t')$ in the plane and the curve $C$ does not go through any intersection point $y_{i,j}$ (informally, the curve is directed to the top and right). 
Let $C$ be any good curve. The ordering $\preceq_C$ is defined on the set of paths $\pset'$ intersecting $C$ as follows: Paths $P \prec_C P'$ if and only if $C$ intersects $P$ before it intersects $P'$.
Since $C$ does not intersect point $y_{i,j}$, either $P \prec_C P'$ or $P' \prec_C P$.

\begin{figure}[hbt]
\begin{center} 
\includegraphics[scale=0.45, clip=true, trim=  0 5cm 5cm 0cm] {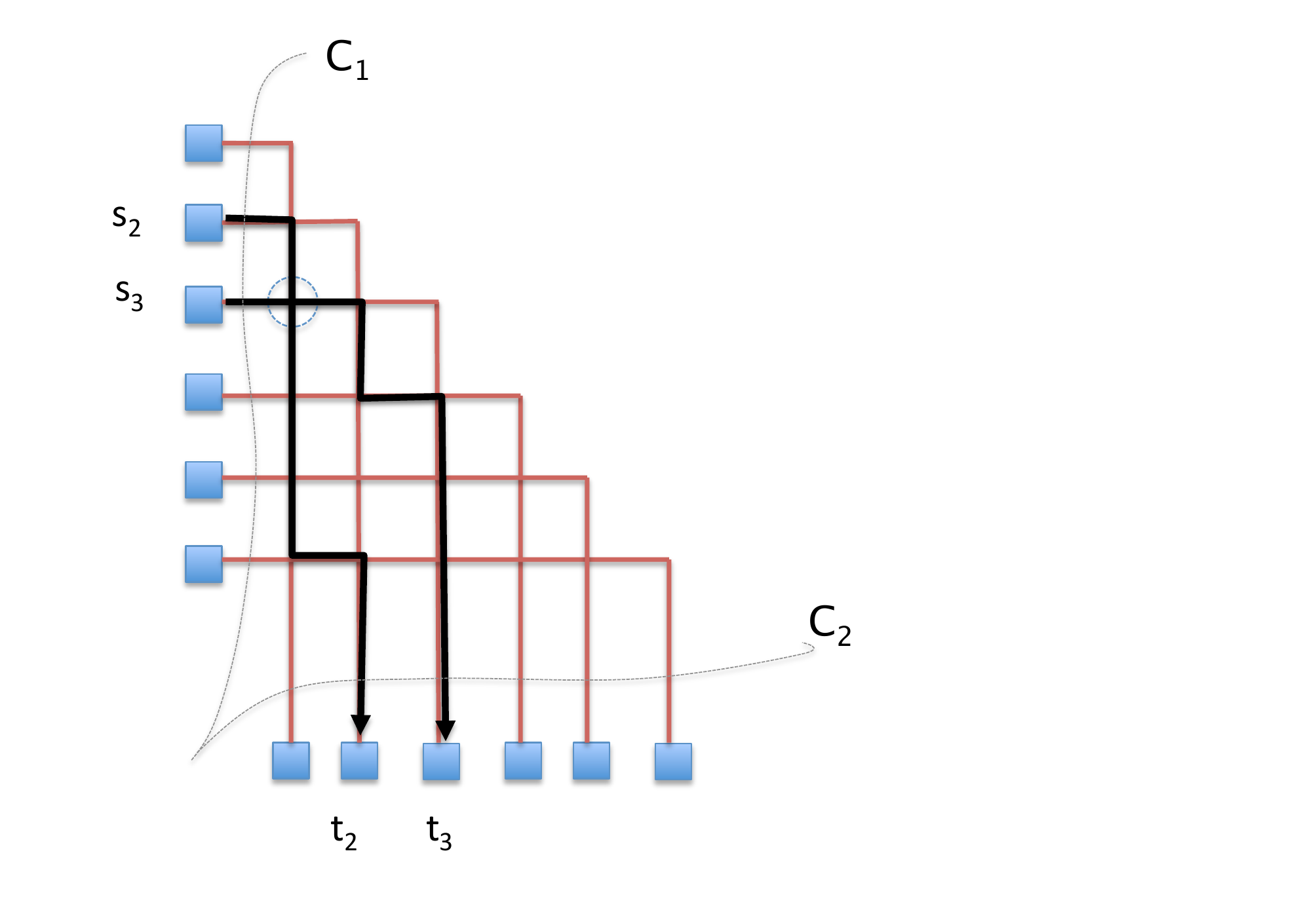}
\caption{Both $C_1$ and $C_2$ are good curves that originated from $(0,0)$ (illustrated by dotted lines). We have $Q' \prec_{C_1} Q$ while $Q \prec_{C_2} Q'$. \label{fig: curve}}  
\end{center} 
\end{figure}

\paragraph{Regions and Switching Boxes.}
In $\wall [G \cdot H]$, we have canonical paths $P_{ia}$ for $i \in V(G)$ and $a \in V(H)$. 
For each $i \in V(G)$, we define a {\em region} $R_i$ on the plane that contains all paths $P_{(i,a)}$ for $a \in [r]$. 
For $i, j \in V(G)$, the intersection between regions $R_i$ and $R_j$ is called a {\em bounding box} $B(i,j)$ which contains $|V(H)|^2$ virtual vertices of the form $Y{(i,a), (j,b)}$ for $a, b \in V(H)$.
Notice that a canonical path $P_{(i,a)}$ is completely contained inside region $R_i$, and as we walk on the path from $s(i,a)$ to $t(i,a)$, we will visit the bounding boxes $B(i,1),\ldots, B(i,n)$ in this order.
For convenience, the region in $R_i$ between $B(i,i-1)$ and $B(i,i+1)$ is called $B(i,i)$.  
See Figure~\ref{fig: region} for illustration.  

\begin{figure}[hbt]
\begin{center} 
\includegraphics[scale=0.5, clip=true, trim=  0 7cm 9cm 0cm] {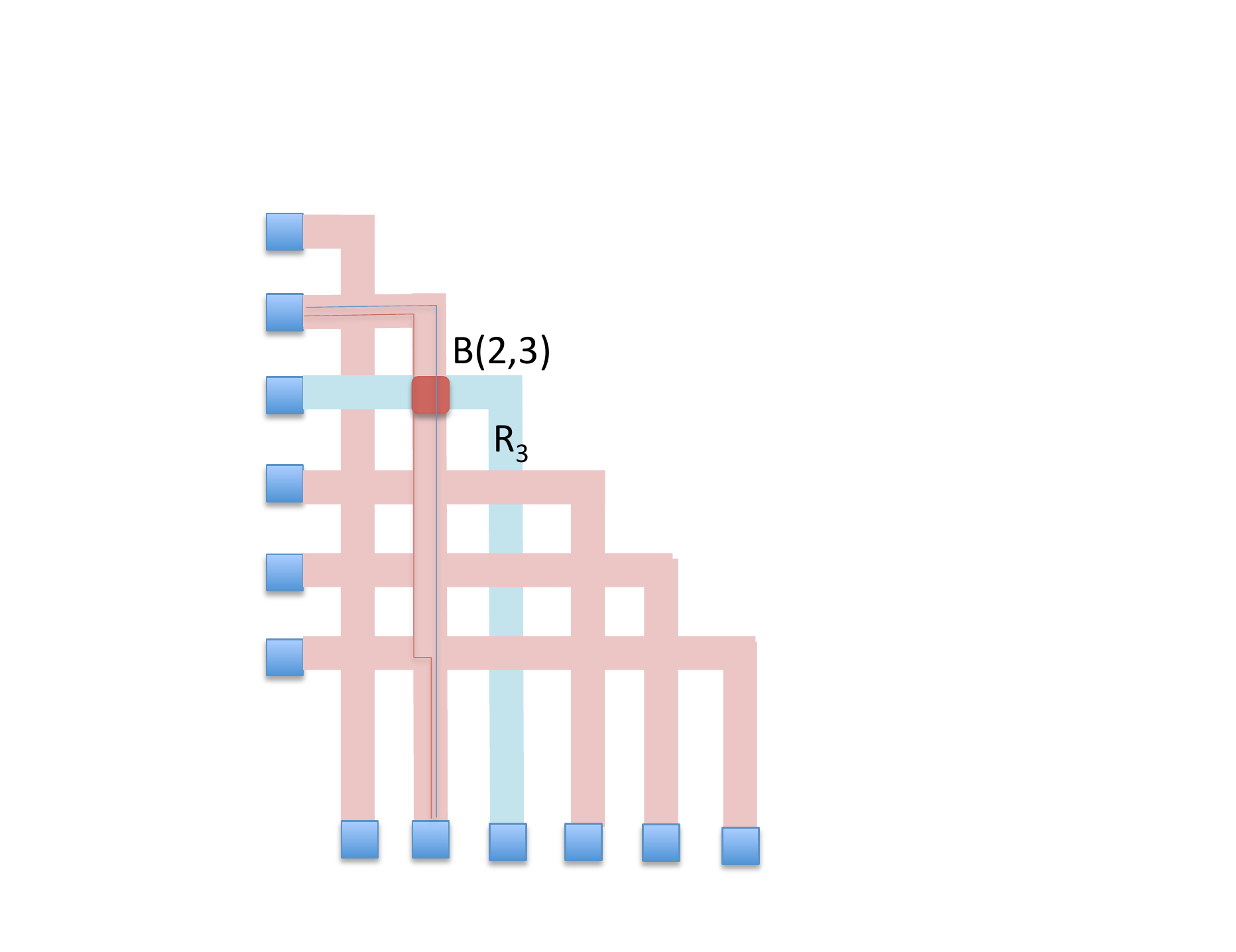}
\caption{Regions and switching boxes in $\vec{R}[G]$. There are two paths routed inside region $R_2$. \label{fig: region}}
\end{center} 
\end{figure}

\begin{proposition} 
Consider any box $B(i,j)$ for $i \neq j$. One of the following two cases holds: 
\begin{itemize} 
\item For all $a,b \in V(H)$, the virtual vertex $y_{(i,a), (j,b)}$ is a directed edge $e_{(i,a),(j,b)}$. This happens when $ij \in E(G)$, and we say that the box $B(i,j)$ is a non-switching box.  

\item For all $a,b \in V(H)$, the virtual vertex is an uncrossing, in which case, we say that the box $B(i,j)$ is a switching box.   

\end{itemize} 
\end{proposition}

The term switching box is coined from an intuitive reason: Consider a switching box $B(i,j)$ and a collection of edge-disjoint paths that are routed in a solution.  
Let $\pset_{top}$ and $\pset_{left}$ be the paths in the solution that enter this box from the top and left respectively, so paths in $\pset_{top}$ (resp. $\pset_{left}$) must leave the box from the bottom (resp. right). 
Define the curves $C_{in}$ (and $C_{out}$) as the union of left and top boundaries of $B(i,j)$ (resp. the union of right and bottom boundaries).
With respect to the curve $C_{in}$ all paths in $\pset_{top}$ are ordered after paths in $\pset_{left}$, while this becomes the opposite for $C_{out}$.
In other words, the box $B(i,j)$ ``switches'' the order of these paths.


\subsection{Proof}\label{sec:proof edp ineq}

Now we prove Lemma~\ref{lem: edp product}. 
Let $I$ be the set of indices of edge-disjoint paths in $\wall [G \cdot H]$ where, for each $(i,a) \in I$, there is a path $Q_{(i,a)}$ connecting $s(i,a)$ to $t(\psi(i,a))$ and the paths $Q_{(i,a)}$ are edge-disjoint and orderly feasible (recall that orderly feasible solutions may connect $s(i,a)$ to some other sink $t(i',a')$).  
We say that a path $Q_{(i,a)}$ is {\em semi-canonical} if it is completely contained in region $R_i$. 
Let $I' \subseteq I$ be the set of semi-canonical paths. 

\begin{lemma} 
$|I'| \leq \alpha(G) \tildeEDP(\wall [H])$ 
\end{lemma}    

\begin{proof}
We first define the partition of $I'$ by the first coordinates of paths.  
Define $I'_i = \set{(i,a): (i,a) \in I'}$, so we will have $I' = \bigcup_{i \in V(G)} I'_i$.  
We count the number of indices $i \in V(G)$ such that $I'_i \neq \emptyset$. Define $\Lambda = \set{i \in V(G): I'_i \neq \emptyset}$. We claim that $\Lambda$ is an independent set and therefore  $|\Lambda| \leq \alpha(G)$: Suppose $i, j \in V(G)$ such that $I'_i, I'_j \neq \emptyset$. 
Let $(i,a) \in I'_i$ and $(j,b) \in I'_j$ be any two paths.  
Due to the fact that this is an orderly feasible solution, these two paths must cross at some virtual vertex $y_{(i,a'), (j,b')}$ inside box $B(i,j)$, and they must share an edge $e_{(i,a'), (j,b')}$, a contradiction.
This implies that $|\Lambda| \leq \alpha(G)$. 

Next, we argue that $|I'_i| \leq \tildeEDP(\wall [H])$ for all $i \in \Lambda$, which will complete the proof of the lemma.
If we consider the box $B(i,i)$, we see an isomorphic copy $H'$ of graph $ \wall [H]$, in which each path $Q_{(i,a)}$ corresponds to another path $Q'_{\phi(a)}$ that connects some ``source'' $s'(\phi(a))$ to ``sink'' $t'(\xi(a))$.
We claim that the collection of paths $Q'_{\phi(a)}$ is orderly feasible in the instance $(H', \set{(s'(a),t'(a)}_{a \in V(H)})$: Assume otherwise that some paths $Q'_{\phi(a)}$ and $Q'_{\phi(b)}$ do not cross inside $H'$, so the paths $Q_{(i,a)}$ and $Q_{(i,b)}$ must cross at some other box $B(i,j)$ for $i \neq j$. 
If such box is a switching box, it is impossible for these two paths to cross because they must enter and leave the box from the same direction; otherwise, if box $B(i,j)$ is not a switching box, it is also impossible for them to cross.   
\end{proof} 

Let $I'' = I \setminus I'$. For convenience, let us renumber $I''$ such that $I'' = \set{1,\ldots, t}$ such that the source of path $1$ is above that of path $2$ and so on.  
Now we will show that $|I''| \leq 3 |E(G)| +1$. 
Our proof will rely on the notion of configurations.
We first define the order of boxes $B(i,j)$ for $i > j$ such that $B(2,1)$ is the first box, which precedes $B(3,1)$ (the second box), and so on. 
More formally, the box $B(i,j)$ precedes $B(i',j')$ if and only if $i < i'$ or $i=i'$ and $j < j'$; in short, this is simply a lexicographic order of boxes.  
This defines a total order over boxes.   
 
We define a number of good curves $C_1,\ldots, C_z$ for $z = {|V(G)| \choose 2}$, where the curve $C_h$ is any good curve such that (i) $C_h(0) = (0,0)$, (ii) $C_h(1) = (x_{\max}, y_{\max})$ and (iii) the first $h-1$ boxes are above $C_h$, while $z - h+1$ curves are below it (notice that $C_h$ partitions the region $[0, n+1] \times [0,n+1]$ into two parts, i.e. one above the curve and the other below it). 

\begin{observation}
For each $h = 1,\ldots, z$ and path $i \in I''$, the curve $C_h$ intersects path $i \in I''$.  
\end{observation}

\begin{proof} 
This is just because any path in the orderly feasible solution starts from the region above the curve $C_h$, while it ends in the region below the curve.  
\end{proof}

For each $h = 1,\ldots, z$, a curve $C_h$ can be used to define a configuration $\sigma_h = (x_1,\ldots, x_t)$ of paths in $I''$ where $x_j \in I''$ is the index of the $j$th path that intersects with the curve $C_h$(this order is well-defined because the curve has directions). 
Notice that $\sigma_1 = (t,\ldots, 1)$, and $\sigma_z = (1,\ldots, t)$.

The number of {\em reversals} of a configuration $\sigma_h$ is the number of locations $j$ such that $\sigma_j > \sigma_{j+1}$. Denote this number by $rev(\sigma_h)$, so we have that $rev(\sigma_1) = t-1$, and $rev(\sigma_z)  =0$.
Our proof proceeds by analyzing how the number of reversals changes over configurations $\sigma_1, \ldots, \sigma_z$. 
We will show that, for any $h$, we have $rev(\sigma_h) - rev(\sigma_{h+1}) \leq 3$, which implies that $t-1=  rev(\sigma_1) - rev(\sigma_z) = \sum_{h=1}^{z-1} (rev(\sigma_h) - rev(\sigma_{h+1})) \leq 3 (z-1)$; in other words, $|I''| = t \leq 3 z \leq 3 |V(G)|^2$.  
So the last thing we need to prove is the following lemma: 

\begin{lemma} 
For any $h=  1,\ldots, z-1$, we have $rev(\sigma_h)  -rev(\sigma_{h+1}) \leq 3$. 
\end{lemma} 

\begin{proof} 
Let $B(i,j)$ be the $h$th box and $J \subseteq I''$ be the indices of paths entering this box.
If the box $B(i,j)$ is a non-switching box, then it must be the case that $\sigma_{h+1} = \sigma_h$due to the fact that paths cannot cross inside region $B(i,j)$.  
This implies that $rev(\sigma_h) = rev(\sigma_{h+1})$ in this case.  

Now we consider the other case when $B(i,j)$ is a switching box. 
We write $J = J_{top} \cup J_{left}$ where $J_{top}$ (and $J_{left}$) is the set of indices of paths entering box $B(i,j)$ from the top (and left respective). It is clear that paths coming out of the bottom and right of the box are exactly $J_{top}$ and $J_{left}$ respectively. 
Notice that, while the curve $C_h$ crosses $J_{top}$ after $J_{left}$, the curve $C_{h+1}$ would cross paths in $J_{left}$ before those in $J_{top}$.  
The configurations $\sigma_h$ and $\sigma_{h+1}$ can be written as $\sigma_h = \sigma' \circ \sigma^{left} \circ \sigma^{top} \circ \sigma''$ and $\sigma_{h+1} = \sigma' \circ \sigma^{top} \circ \sigma^{left} \circ \sigma''$ respectively.  
See Figure~\ref{fig: cross} for illustration.  
\end{proof} 

\begin{figure}[hbt]
\begin{center} 
\includegraphics[scale=0.9, clip=true, trim=  0 10cm 9cm 3cm] {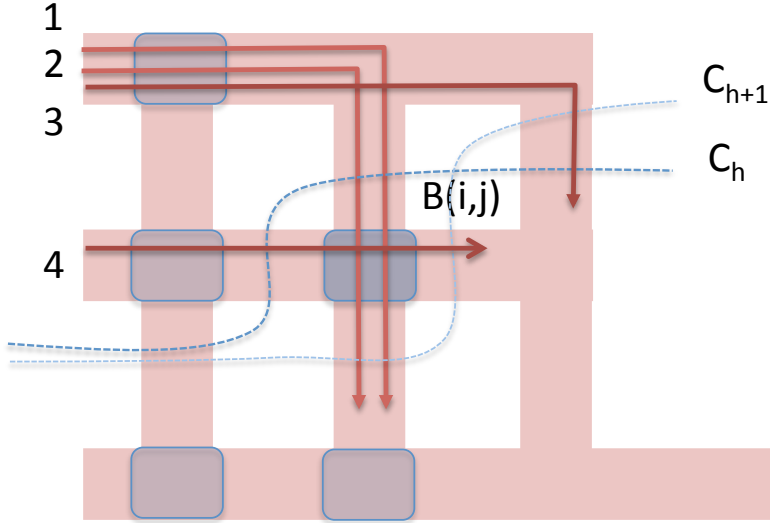}
\caption{Configurations of curve $C_h$ and $C_{h+1}$ are $\sigma_h = (4,2,1,3)$ and $\sigma_{h+1} = (2,1,4,3)$ respectively. In this case, $\sigma^{top} = (2,1)$ and $\sigma^{left} = (4)$. Box $B(i,j)$ is a switching box.\label{fig: cross}}
\end{center} 
\end{figure}

%% file: others.tex
\section{Other Problems}\label{sec:other}

In this section, we prove the hardness of $k$-Cycle Packing and DNF/CNF Minimization. 
As noted previously, our proof for CNF minimization is an alternative proof of Aleknovich et al. \cite{AlekhnovichBFKP08}. 

\subsection{Hardness of $k$-Cycle Packing for Large $k$}

We consider the problem of packing edge-disjoint $k$-cycles in which our goal is to pack as many cycles of length at most $k$ as possible. 
We only need to slightly change the reduction $\wall [G]$ as used in Section~\ref{sec:edp} in the following way: In the second step, for each pair $i, j \in V(G)$, if $ij \in E(G)$, we do the same, but for $ij \not \in E(G)$ (including the case when $i=j$), we make two new vertices on each line before and after the jump (see Figure~\ref{fig: edc}).
Also, we have a {\em back edge} from $t_i$ to $s_i$ for each $i \in V(G)$.  

\begin{figure}[hbt]
\begin{center} 
\includegraphics[scale=0.5, clip=true, trim=  0 3cm 4cm 0cm] {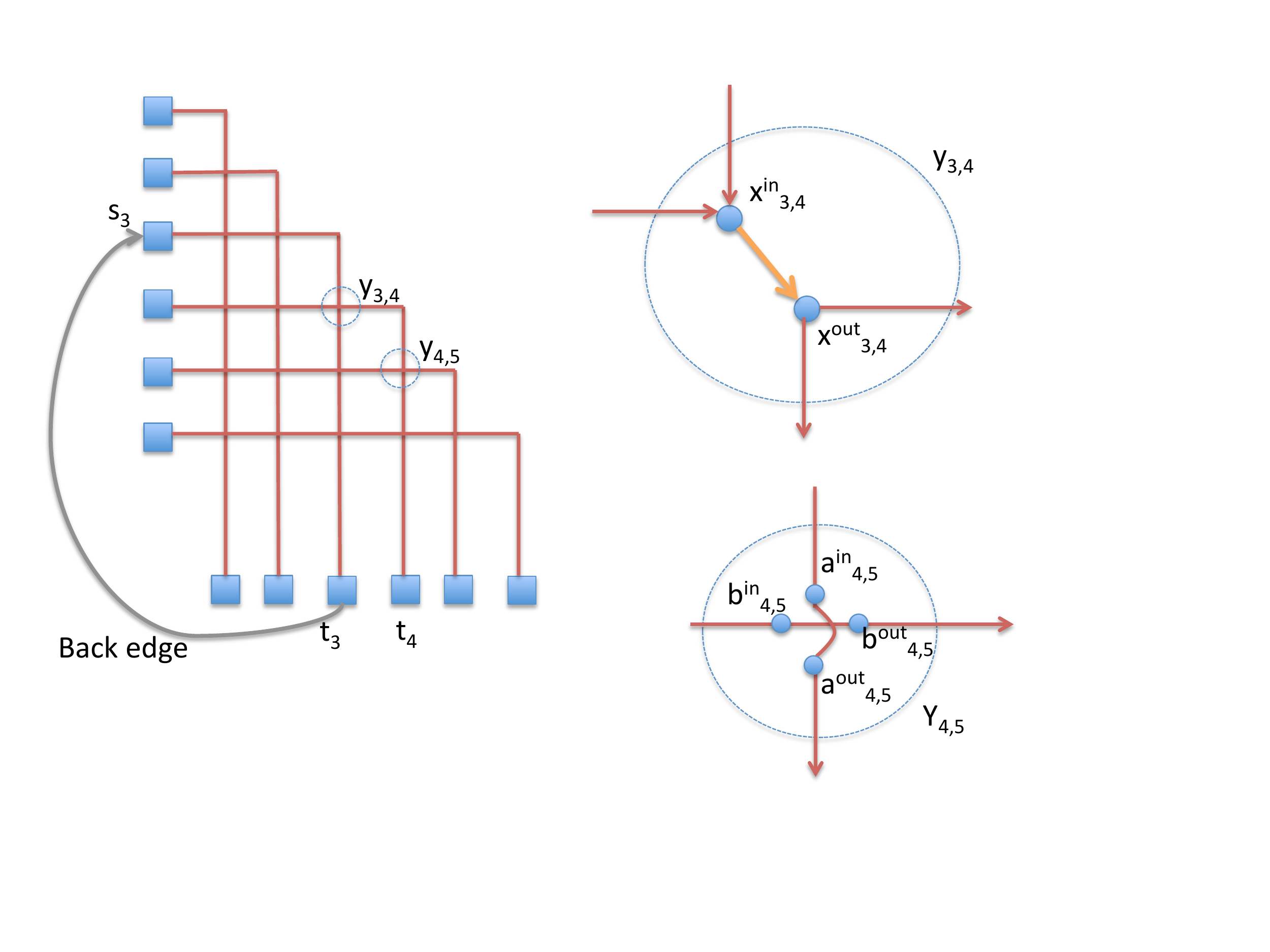}
\caption{The graph $\protect\wall[G]$ where $(3,4) \in E(G)$ but $(4,5) \not\in E(G)$. The differences between this gadget and EDP gadget only lies in the new vertices $a^{in}_{i,j}, a^{out}_{i,j}, b^{in}_{i,j}, b^{out}_{i,j}$}  
\label{fig: edc}
\end{center} 
\end{figure}

With this reduction, any ``canonical'' cycle between source $s_i$ to sink $t_i$ (and taking back edge to $s_i$) must have length exactly $2n+2$, so we choose the value $k=2n+2$.
Let $edc(\wall' [G])$ denote the optimal value of $k$-cycle packing. 
We now establish the connection between the optimal value of EDP solution in $\wall [G]$ and the $k$-EDC solution in $\wall' [G]$. 

Notice that for any cycle that uses only one back edge, there is a corresponding path from some $s_i$ to $t_i$. 
The number of these cycles corresponds exactly to $\EDP(\wall [G])$, so we can write 
\[\EDP(\wall [G]) \leq  edc(\wall' [G]) \leq \EDP(\wall [G]) + \widetilde{edc}(\wall' [G])\] 
where $\widetilde{edc}(\wall' [G])$ is the number of cycles that use more than one back edge. 
The following lemma says that these cycles must be longer than $k$, i.e. $\widetilde{edc}(\wall' [G]) = 0$. In other words, this implies that $edc(\wall' [G]) = \EDP (\wall [G])$.  
  

\begin{lemma} 
Let $C$ be a cycle in $\wall' [G]$ that uses more than one back edge. Then $|C| >k$. 
\end{lemma} 
\begin{proof} 
Any cycle $C$ must start at some $s_i$ and ends at $s_i$. 
Let $i = i_0, i_1, \ldots, i_{\ell} = i$ be the indices of the source-sink pairs visited in the cycle $C$, i.e. the cycle goes through $s_{i_0} \rightarrow t_{i_1} s_{i_1} \rightarrow s_{i_2} \rightarrow \ldots \rightarrow t_{i_{\ell}} s_{i_{\ell}}$.
Observe that a path that goes from $s_j$ to $t_{j'}$ visit exactly $(n  - j +j')$ vertices of the form $y_{i,j}$ , because such path must go right $j'$ times and go down $n-j$ times (in arbitrary order). 
Combining this with $s_j$ and $t_{j'}$, such path would visit $2(n-j+j'+1)$ vertices.  
Therefore, the total length of the cycle $C$ is  
\[\sum_{j=0}^{\ell-1} 2(n-i_j + i_{j+1}+1) = 2 \ell (n+1) + 2(i_{\ell} - i_0) = 2 \ell (n+1)   \] 
So this cycle would have been longer than the threshold $k$ if $\ell >1$.   
\end{proof}

\subsection{Learning CNF Formula}

We present an alternative proof for the hardness of properly learning CNF using our framework.  
Our reduction is quite similar to Alekhnovich~et~al.'s (see \cite{AlekhnovichBFKP08}), but our proof highlights the role of graph products in the proof (while their construction cannot be seen as a standard graph product in any way).

Let $G$ be any graph. We think of a vertex $u\in V(G)$ as an integer in $\set{1,\ldots, n}$. 
For each vertex $u \in V(G)$, we define an encoding $\enc{u} = 0^{u-1} 1 0^{n-u}$. 
For each edge $uv \in V(G)$, the encoding of an edge $uv$ has two $1$s at the positions corresponding to $u$ and $v$.  
Our reduction encodes the $k$-fold graph product $H=G^k$ into samples as follows.
For each $\vec{u}= (u_1,\ldots, u_k) \in V(H)$, we define a negative sample $neg(\vec{u}) = \enc{u_1} \ldots \enc{u_k}$. 
For each $\vec{u} \in E(H), i \in [k]$ and $u_i v \in E(G)$, we define a positive sample $pos(\vec{u}, v,i ) = \enc{u_1} \ldots \enc{u_{i-1}} \enc{u_iv} \enc{u_{i+1}} \ldots \enc{u_k}$.  

Notice that the total number of variables is $n k$, where we think of them as $k$ blocks; in each of which, there are $n$ variables. 
Denote by $z(i,u)$ the variable in block $i \in [k]$ that corresponds to a vertex $u \in V(G)$.

\begin{lemma} \label{thm:CNF lower}
$\opt(R[H]) \geq \chi(H)$
\end{lemma} 

\begin{proof} 
Suppose $\bigwedge_{j=1}^M C_j$ be a CNF formular that is consistent with all samples. 
We claim that the number of clauses $M$ is at least $\chi(H)$. 
For each $\vec{u} \in V(H)$, let $\sigma(\vec{u})$ be the index such that $C_{\sigma(\vec{u})}$ evaluates to false on sample $neg(\vec{u})$; this clause must exist since this is a negative sample (if there are many such indices for vertex $\vec{u}$, we choose any arbitrary one).
Now for each $s \in [M]$, we define the set of vertices $Q_{s} \subseteq V(H)$ as $Q_{s} = \set{\vec{u}: \sigma(\vec{u}) = s}$.  

\begin{claim} 
$Q_{s}$ is an independent set for all $s \in [M]$.  
\end{claim} 
\begin{proof} 
Assume otherwise that some $\vec{u} \vec{v} \in E(H)$ such that $\vec{u}, \vec{v} \in Q_{s}$.
Let $i$ be the index such that $u_j = v_j$ for all $j <i$ and $u_i v_i \in E(G)$.
Let $X, Y \subseteq [k] \times V(G)$ be the subset of variable indices that appear positively and negatively in clause $C_s$, so we can rewrite $C_s = \left( \bigvee_{(j,u) \in X} z(j,u) \right) \vee \left( \bigvee_{(j,u) \in Y} \overline{z(j,u)} \right)$.  
Since $C_{s}$ evaluates to false on both $neg(\vec{u})$ and $neg(\vec{v})$, we can neither have variable $z(i,u_i)$ nor $z(i,v_i)$ in the clause $C_{s}$: Suppose otherwise that $(i,u_i) \in X$ or $(i,u_i) \in Y$, then either $neg(\vec{u})$ or $neg(\vec{v})$ would have evaluated to true on clause $C_s$ (contradicting $\vec{u} \in Q_s$). 
Similarly, if $(i,v_i) \in X$ or $(i,v_i) \in Y$, then either $neg(\vec{v})$ or $neg(\vec{u})$ would have been true in clause $C_s$.

In other words, $(i,u), (i,v) \not \in X \cup Y$. 
But notice that $pos(\vec{u}, v_i,i)_{(i',u')} = neg(\vec{u})_{(i',u')}$ for all $(i',u') \not \in \set{(i,u_i), (i,v_i)}$, so $C_s$ must evaluate to false on input $pos(\vec{u}, v_i, i)$, a contradiction.  
%
\end{proof}
We have just shown that $\set{Q_s}_{s \in [M]}$ is a valid $M$-coloring for graph $H$, so we must have $M \geq \chi(H)$, as desired.  
\end{proof} 

Next, we prove the upper bound.

\begin{lemma} \label {lem:CNF-upperbound}
$\opt(R[H]) \leq \chi(G)^k n^{O(1)}$ 
\end{lemma}
\begin{proof}
We construct the same formula as in Aleknovich et al.  
That is, let $I_1,\ldots, I_M$ be color classes of $G$ and $\sigma: V(G) \rightarrow [M]$ be the corresponding coloring function.  
Define the formula $f_i(z)  = \wedge_{c=1}^M \vee_{u \not \in I_c} z(i,u) $,
for $i=1,2,\ldots,k$,
and define $f(z) = \vee_{i=1}^k f_i(z)$. 
This formula can be turned into a CNF of size at most $\chi(G)^k |V(G)|^{O(1)}$.  

\begin{claim} \label{clm:CNF-is-consistent}
The formula $f$ is consistent with all the samples. 
\end{claim} 

\begin{proof} 
Consider each negative sample $neg(\vec{u})$ for $\vec{u} = (u_1,\ldots, u_k) \in V(H)$. 
For each $i \in [k]$, notice that $f_i(\enc{u_i})$ evaluates to false because $\vee_{u \not \in I_{\sigma(u_i)}} \enc{u_i}_{ u}$ is false (since $u_i \in I_{\sigma(u_i)}$ is the only bit of $\enc{u_i}$ that is ``1''). 
This implies that $\vee_{i=1}^k f_i(neg(\vec{u}))$ is false. 

Now consider, for each $\vec{u} \in V(H)$, $v: u_i v \in E(G)$ and $i \in [k]$, a positive sample $pos(\vec{u}, v, i)$.  
We claim that $f_i(pos(\vec{u}, v,i))$ is true, which causes $f(pos(\vec{u}, v, i))$ to be true: Assume to the contrary that $f_i$ is false, so some term $\vee_{w \not \in I_c} \enc{u_i v}_{w}$ is false for some $c$; notice that it can be false only if $\enc{u_i v}_w = 0$ for all $w \not \in I_c$; since we have $\enc{u_i v}_{u_i} = \enc{u_i v}_{v} = 1$, it must be the case that both $u_i$ and $v$ belong to $I_c$, contradicting the fact that $I_c$ is a color class.     
\end{proof} 

\end{proof}

%% file: conclusion.tex
\section{Conclusion and Open Problems}

We have shown applications of pre-reduction graph product techniques in proving hardness of approximation.
For some applications, such as EDP, proving $\alpha$-projection property implies tight hardness, but for some others, we need a more careful reduction of the form $R[G^{\ell}]$ (taking into account the fact that the input is an $\ell$-fold product of graphs).  

There are many open problems on edge-disjoint paths. 
Most notably can one narrow down the gap of undirected EDP between $O(\sqrt{n})$ upper bound and $\log^{1/2-\epsilon} n$ lower bound?  
For directed EDP, there is still a (relatively large) gap in the case of low congestion routing, between the upper bound of $n^{1/c}$~\cite{KolliopoulosS04} and the lower bound of $n^{1/(3c+23)}$~\cite{ChuzhoyGKT07} if we allow routing with congestion $c$. 
We believe that our techniques are likely to work there (in a much more sophisticated manner), and it would potentially close this gap. 
This would resolve an open question in Chuzhoy et al.~\cite{ChuzhoyGKT07}.

Another interesting problem is the cycle packing problem. 
For this problem, the approximability is pretty much settled on undirected graphs with an upper bound of $O(\log^{1/2} n)$ and a lower bound of $\log^{1/2 -\epsilon} n$~\cite{FriggstadS11,KrivelevichNSYY07}.
On directed graphs, there is still a large gap between $n^{1/2}$ and $\Omega(\frac{\log n}{\log \log n})$.
For $k$-cycle packing problem, it is interesting to see whether our technique gives $k^{1-\epsilon}$ hardness for small $k$.

\paragraph{Acknowledgement:}
We thank Julia Chuzhoy for suggesting the EDP reduction and for related discussions when the first author was still at the University of Chicago.

%% file: bad-examples.tex
\section{List of Bad Examples}
\label{sec: bad}

In this section, we provide the evidences that all the reductions considered in this paper are neither trivially ``working''\danupon{What is ``working''?} nor subadditive in the sense of our previous SODA paper \cite{ChalermsookLN-SODA13}. 
Therefore, we will need the new conceptual ideas introduced in this paper.

\danupon{If there will be only one subsection, we should remove the subsection.}




\subsection{Edge Disjoint Paths}\label{sec:bad example edp}

\danupon{It's still very unclear why this is an evidence that we cannot use our SODA concept}

We show a graph $G$ in which $\alpha(G) = 2$ but $\EDP(\wall [G]) = n/3$.
To ensure that $G$ does not have an independent set of size $3$, we define graph $G$ by defining a triangle-free graph $H$ and let $G= K_n \setminus H$. 
We consider a set of vertices $V(G)  = A \cup B \cup C$, where $|A|= \set{1,\ldots, n/3}, B= \set{n/3+1,\ldots, 2n/3}$ and $C= \set{2n/3 + 1, \ldots, n}$.
Graph $H$ only have edges between $A$ and $B$ in such a way that, for any $i \in A$ and $j \in B$, we have $ij \in E(H)$ if and only if $i -j > n/3$. 
It is obvious that $H$ is bipartite, so if we define $G=  K_n \setminus H$, we have $\alpha(G) =2$.

\begin{figure}[hbt]
\begin{center} 
\includegraphics[scale=0.6, clip=true, trim=  0 7cm 5cm 3.4cm] {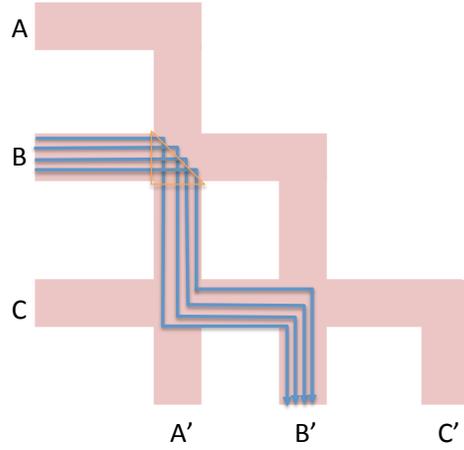}
\caption{Bad Example for EDP reduction}
\label{fig: bad edp}
\end{center} 
\end{figure}

Now we check that $\EDP(\wall [G]) = n/3$. For each $i \in [n]$, let $E_i$ denote the set of edges on the canonical path from $s_i$ to $t_i$ in $\wall [G]$.
For each $i \in B$, we define a path $P_i$ that: 
\begin{itemize} 
\item Start at $s_i$ and go straight until it meets with an edge in $E_{i- n/3}$, at which the path turns downward (the turning is possible because we have an edge $i (i-n/3) \in E(G)$). 

\item The path goes downward until it meets with an edge in $E_{n - i +n/3}$, at which the path turns again towards the right.

\item Once path $P_i$ meets with an edge in $E_i$ again, it takes a turn downward and remains so until it reaches $t_i$.   
\end{itemize}  

\begin{observation} 
For any $i ,j \in B$, $P_i \cap P_j = \emptyset$. 
\end{observation}


%
%
%
%

\danupon{I commented out the DNF minimization section since there's nothing.}